\pdfoutput=1
\documentclass[letterpaper,11pt]{article}
\usepackage[margin=1in]{geometry}
\usepackage{amsmath, amsfonts, amssymb, amsthm}
\usepackage{thm-restate}
\usepackage[colorlinks=true,urlcolor=black,linkcolor=black,citecolor=black]{hyperref}
\usepackage[capitalize]{cleveref}
\usepackage{bbm,url}
\usepackage{graphicx}
\usepackage{xcolor}
\usepackage{colortbl}
\usepackage{subcaption}
\usepackage{hhline}
\usepackage{pifont}
\usepackage{multirow}
\usepackage{multicol}
\usepackage{enumitem}
\usepackage{comment}
\usepackage{nicefrac,bm}
\usepackage{algorithm}
\usepackage[noend]{algpseudocode}
\usepackage{thmtools}
\usepackage{thm-restate}

\allowdisplaybreaks

\newtheorem{theorem}{Theorem}
\newtheorem*{theorem*}{Theorem}

\newtheorem{lemma}{Lemma}
\newtheorem{claim}{Claim}

\newtheorem{example}[theorem]{Example}

\theoremstyle{definition}
\newtheorem{definition}{Definition}

\newcommand{\hide}[1]{}

\newcommand{\s}[1]{\mathsf{#1}}

\newcommand{\mpb}{\mathsf{MPB}}

\newcommand{\A}{\mathcal{A}}

\newcommand{\x}{\mathbf{x}}
\newcommand{\y}{\mathbf{y}}
\newcommand{\z}{\mathbf{z}}
\newcommand{\p}{\mathbf{p}}

\newcommand{\N}{\mathbb N}

\newcommand{\R}{\mathbb R}
\newcommand{\poly}[1]{\mathsf{poly}(#1)}

\newcommand{\wps}{\mathsf{WPS}}

\title{Weighted EF1 and PO Allocations with \\Few Types of Agents or Chores\thanks{Work supported by NSF Grants CCF-1942321 and CCF-2334461.}} 

\author{Jugal Garg\footnote{University of Illinois at Urbana-Champaign, USA} \\
\texttt{\small jugal@illinois.edu} 
\and
Aniket Murhekar\footnote{University of Illinois at Urbana-Champaign, USA}\\
\texttt{\small aniket2@illinois.edu}
\and
John Qin\footnote{University of Illinois at Urbana-Champaign, USA}\\
\texttt{\small johnqin2@illinois.edu}
}

\date{}
\begin{document}
\renewcommand{\arraystretch}{1.2}
\maketitle

\begin{abstract}
We investigate the existence of fair and efficient allocations of indivisible chores to asymmetric agents who have unequal entitlements or weights. We consider the fairness notion of weighted envy-freeness up to one chore (wEF1) and the efficiency notion of Pareto-optimality (PO). The existence of EF1 and PO allocations of chores to symmetric agents is a major open problem in discrete fair division, and positive results are known only for certain structured instances. In this paper, we study this problem for a more general setting of asymmetric agents and show that an allocation that is wEF1 and PO exists and can be computed in polynomial time for instances with:
\begin{itemize}[leftmargin=*]
\item Three types of agents, where agents with the same type have identical preferences but can have different weights. 
\item Two types of chores, where the chores can be partitioned into two sets, each containing copies of the same chore.
\end{itemize}
For symmetric agents, our results establish that EF1 and PO allocations exist for three types of agents and also generalize known results for three agents \cite{GMQ23chores}, two types of agents \cite{GMQ23chores}, and two types of chores \cite{aziz2022twotypes}. Our algorithms use a weighted picking sequence algorithm as a subroutine; we expect this idea and our analysis to be of independent interest. 
\end{abstract}

\section{Introduction}
Fair division is a ubiquitous problem in many disciplines, such as computer science, economics, social choice, and multi-agent systems. While the formal study began with the cake-cutting problem proposed by Steinhaus \cite{steinhaus}, which concerned the fair division of a divisible good, the fair division of indivisible items has received considerable attention in recent times (see excellent surveys \cite{aziz2022survey,amanatidis2023survey}). Although fairness of an allocation is inherently subjective, envy-freeness (EF) is a quintessential and well-established notion of fairness \cite{foleyEF}. Envy-freeness requires that every agent (weakly-)prefers the items allocated to her to those allocated to others. Unfortunately, EF allocations need not exist when items are indivisible, as can be seen from the simple example of assigning one task to two agents. Hence, relaxations of EF have been defined to qualify fairness in the discrete case, with envy-free up one item (EF1) being one such popular relaxation. When agent preferences are monotone, EF1 allocations always exist and can be computed in polynomial time \cite{Lipton04onapproximately,bhaskar2020chores}.

\begin{table*}[t]
\centering
\begin{tabular}{|c||c|c|}\hline
Instance & EF1 + fPO & wEF1 + fPO \\\hhline{|=#=|=|}
General Additive & ? & ? \\\hline
Two agents & \checkmark\cite{aziz2019chores} & \checkmark\cite{wu2023wef1} \\\hline
Three agents & \checkmark\cite{GMQ23chores} & \cellcolor{blue!25} \checkmark\cref{thm:three-agent-types} \\\hline
Two-agent-types & \checkmark\cite{GMQ23chores} & \cellcolor{blue!25} \checkmark\cref{thm:three-agent-types} \\\hline
Three-agent-types & \cellcolor{blue!25} \checkmark\cref{thm:three-agent-types} & \cellcolor{blue!25} \checkmark\cref{thm:three-agent-types} \\\hline
Two-chore-types & \checkmark\cite{aziz2022twotypes} & \cellcolor{blue!25} \checkmark\cref{thm:two-chore-types} \\\hline
Bivalued & \checkmark\cite{Garg_Murhekar_Qin_2022,ebadian2021bivaluedchores} & \checkmark\cite{wu2023wef1} \\\hline
\end{tabular}
\caption{State-of-the-art for EF1/wEF1+fPO allocation of indivisible chores. \checkmark~denotes existence/polynomial-time algorithm, ? denotes (non-)existence is unknown. Colored cells highlight our results.}\label{tab:results}
\end{table*}

A fair allocation can be quite sub-optimal in terms of overall efficiency: a set of tasks could be allocated so that every agent is assigned tasks they are bad at; as a result, agents may not envy each other by much, but the allocation is inefficient. It is, therefore, natural to seek allocations that are both fair as well as efficient. The classic notion of economic efficiency is Pareto-optimality (PO): an allocation is PO if there is no re-allocation such that every agent weakly prefers the re-allocation while some agents strictly prefer it. Allocations that are simultaneously EF1 and PO are thus highly desirable, and many works have investigated the existence and fast computation of such allocations in various settings, e.g., \cite{caragiannis16nsw-ef1,Barman18FFEA,GargM21,ebadian2021bivaluedchores,Garg_Murhekar_Qin_2022,GMQ23chores,aziz2019chores,aziz2022twotypes}. A common assumption in these works, including ours, is that agent preferences are additive, i.e., the value to an agent from a set of items is the sum of the values of items in the set. Under additive preferences, merely checking if an allocation is PO is a co-NP-hard problem. This adds to the challenge of investigating the existence and computation of EF1 and PO allocations.

Further, the difficulty of the problem is significantly influenced by the nature of the items, as is the definition of EF1. When items are goods and provide value to the agents receiving them, in an EF1 allocation, every agent prefers her bundle to that of another after the removal of one good from the other agent's bundle. For goods, an EF1 and PO allocation is known to exist and can be computed in pseudo-polynomial time. In a pivotal paper, \cite{caragiannis16nsw-ef1} proved that an allocation with the highest Nash welfare — product of the agents' utilities — is both EF1 and PO. This approach does not lead to fast computation since computing an allocation with maximum Nash welfare is APX-hard. Remarkably, \cite{Barman18FFEA} designed a pseudo-polynomial time algorithm for this problem using the idea of a competitive equilibrium. In this approach, agents are endowed with a fictitious amount of money, the goods are assigned prices, and each agent is allocated goods that give them ‘maximum value-for-money’. The latter ensures that the allocation is fractionally PO (fPO), an efficiency property stronger than PO. While the allocation is not EF1, they transfer goods between agents to reduce envy and make appropriate price changes to maintain PO. Using involved potential function arguments, they prove their algorithm terminates with an EF1 and PO allocation. Later, \cite{GargM21} showed an EF1 and fPO allocation can be computed in pseudo-polynomial time and in polynomial time when agents have a constant number of different values for the goods. Despite these results, designing a polynomial time algorithm remains an open problem. 

On the other hand, when items are chores and impose a cost on agents receiving them, in an EF1 allocation every agent prefers her bundle to that of another after the removal of one chore from her bundle. For chores, even the existence of an EF1 and PO allocation is unclear, let alone computation. At first sight, it may seem that the case of goods and chores are similar, and techniques from the goods setting should be directly adaptable to chores. However, this does not seem to be the case. Firstly, a welfare function like Nash welfare guaranteeing EF1 is not known for chores. Secondly, while the competitive equilibrium approach is promising since it guarantees PO, showing that algorithms using this approach terminate has proved to be challenging. As a result, the existence and computation of an EF1 and PO allocation of chores remains a challenging open problem in discrete fair division. To understand the ‘source of hardness’ of this problem, a series of works have focused on identifying structured instances where the problem becomes tractable, i.e., where EF1 and PO allocations exist. These classes include instances with (i) identical agents (folklore), (ii) two agents \cite{aziz2019chores}, (iii) bivalued disutilities \cite{Garg_Murhekar_Qin_2022,ebadian2021bivaluedchores}, (iv) three agents \cite{GMQ23chores}, (v) two types of agents \cite{GMQ23chores}, and (vi) two types of chores \cite{aziz2022twotypes}. While the computation of EF1 and PO allocations is fairly straightforward for (i) and (ii) using ideas for goods, that for classes (iii) - (vi) is non-trivial. These results follow from carefully designed algorithms, all of which use the competitive equilibrium framework but require involved potential function arguments to prove termination. Table \ref{tab:results} lists these results.

All of the above works assume that agents have equal entitlements, capacities, or stakes. This is a limiting assumption, as practical scenarios may involve agents with different capacities for handling workload. For example, Alice, a teaching faculty at a university, may agree to teach twice as many courses as Bob, a research faculty. Thus, Alice will only envy Bob if she feels the workload from the courses assigned to her is more than twice that of Bob. Likewise, the dissolution of a partnership poses the problem of fairly dividing liabilities; in this context, it is only natural that they be divided according to the entitlements/shares of the partners. These scenarios motivate the definition of weighted envy-freeness (wEF) and its relaxation weighted envy-freeness up to one item (wEF1) in the discrete case. Naturally, like the unweighted case, one seeks fair and efficient allocations, which leads us to the existence and computation of a wEF1 and PO allocation of chores. This question is interesting not only because it generalizes an important open problem in discrete fair division but also because it naturally models practical chore division scenarios. 

For goods, \cite{chak2020wef1} showed that a wEF1 (without PO) allocation can be computed via a \textit{weighted picking sequence algorithm}. This algorithm proceeds in $m$ rounds until all $m$ goods are allocated. In each round, a particular agent is chosen who picks her favorite good among the remaining goods. For chores, the existence of wEF1 allocations was only recently shown \cite{wu2023wef1} via a modification of the weighted picking sequence algorithm, and develops a novel analysis technique. \cite{chak2020wef1} showed that wEF1 and PO allocation exists and can be computed in pseudo-polynomial time for goods by adapting the algorithm of \cite{Barman18FFEA} to the weighted case. For chores, \cite{wu2023wef1} also showed that wEF1 and PO allocation can be computed for two agents and for bivalued chores (classes (ii), (iii) above). The existence of wEF1 and PO allocations for chores remained an open problem, including for the subclasses (iv)-(vi).

\subsection{Our Contributions}
In this work, we study the existence and computation of wEF1 and PO allocations of chores to agents with unequal entitlements. We prove that a wEF1 and fPO allocation exists and can be computed in polynomial time for instances with
\begin{itemize}[leftmargin=*]
\item Three types of agents (Theorem 1). In a \textit{three-agent-type} instance, the disutility function of an agent is one of three given functions. 
\item Two types of chores (Theorem 2). In a \textit{two-chore-type} instance, the chores can be partitioned into two sets, each containing copies of the same chore. 
\end{itemize}

\paragraph{Significance.} Theorem 1 establishes another class for which EF1 and PO allocations exist, namely three-agent-type instances. Theorem 1 also subsumes previous works showing the existence of EF1 and PO allocations for (i) two agents \cite{aziz2019chores}, (ii) three agents \cite{GMQ23chores}, and (iii) two-agent-types \cite{GMQ23chores}. We note that improving the existence from three agents to three types of agents in the symmetric setting is itself significant, e.g., for another popular fairness notion of \emph{maximin share (MMS)}, MMS allocations exist for two agents but not for two types of agents~\cite{FeigeST21}. Theorem 2 subsumes previous work showing EF1 and PO allocations exist for two-chore-type instances \cite{aziz2022twotypes}, and also provides an alternative algorithm for the symmetric case. Moreover, our ideas can also be applied to the goods setting to show that wEF1 and PO allocations can be computed in polynomial time for these classes. We record our results in Table \ref{tab:results}.

Our results can also be of practical significance in certain settings. For instance, a cluster may have machines with three types of processing power, e.g., CPUs, GPUs, and TPUs (three types of agents). Likewise, allocation problems may involve jobs that are either heavy or light (two types of chores).

\paragraph{Techniques.} We first remark that simply replacing an agent with as many copies of the agent as their weight, computing an EF1 and PO allocation, and combining the allocation of all the copies does not guarantee a wEF1 and PO allocation; see Appendix \ref{app:examples} for an example. Our algorithms use the competitive equilibrium framework to ensure fPO (hence PO). To obtain wEF1, chores are transferred from one set of agents to another while performing appropriate changes to the chore payments\footnote{Chores have attached payments while goods have prices.} to maintain a competitive allocation. However, these transfers and payment changes are carefully and selectively performed so that we can guarantee the termination of our algorithms. 

We first design \cref{alg:three-agent-types} to compute a wEF1 and fPO allocation for three-agent-type instances using a novel combination of weighted picking sequence and competitive equilibrium framework. We group agents according to their type and allocate chores to agents in a group using a ‘weighted picking sequence’ algorithm. 
We begin by allocating all chores to one group and transferring chores away from this group while ensuring that agents in the other two groups do not wEF1-envy each other. For fPO, we carefully maintain the allocation at a competitive equilibrium throughout the algorithm.

We next design \cref{alg:two-chore-types} to compute a wEF1 and fPO allocation for two-chore-type instances. We first observe that fPO allocations in two-chore-type instances follow a certain ‘ordered’ structure \cite{aziz2022twotypes}. Leveraging this idea, we initially allocate all chores to one agent called the pivot and repeatedly transfer a chore away from the pivot. These transfers respect the ordered structure ensuring fPO and are performed until we eventually obtain a wEF1 allocation, or conclude that the initial choice of the pivot is incorrect. We argue that there exists some choice of the pivot for which the algorithm results in a wEF1 and fPO allocation.

The correctness and termination of our algorithms rely on several involved and novel potential function arguments. In particular, both our results crucially utilize properties of a \textit{weighted picking sequence} algorithm for chores (stated and proved in Lemmas \ref{lem:wps-transfer-1}, \ref{lem:wps-transfer-2}, \ref{lem:two-chore-type-wps}). We believe this may be of independent interest, and may find use in algorithm design for fair division to asymmetric agents in the future. 

\section{Preliminaries}\label{sec:prelim}
An instance $(N, W, M, D)$ of the fair division problem with chores consists of a set $N = [n]$ of $n$ agents, a list $W = \{w_i\}_{i\in N}$ with $w_i > 0$ denoting the weight of agent $i$, a set $M = [m]$ of $m$ indivisible chores, and a list $D = \{d_i\}_{i \in N}$, where $d_i : 2^M \rightarrow \R_{\ge 0}$ is agent $i$'s \textit{disutility} function over the chores. We let $d_i(j)$ denote the disutility of chore $j$ for agent $i$. We assume disutility functions are additive, so that for every $i \in N$ and $S \subseteq M$, $d_i(S) = \sum_{j \in S} d_i(j)$. We consider the following structured classes. An instance $(N,W,M,D)$ is said to be a:
\begin{itemize}[leftmargin=*]
\item \textit{$k$-agent-type} instance if there are $k$ \textit{types} of agents. That is, there is a set $C = \{c_1,\dots,c_k\}$ of $k\in \N$ disutility functions  s.t. for all $i\in N$, $d_i\in C$.
\item \textit{$k$-chore-type} instance if there are $k$ \textit{types} of chores. That is, the set of chores can be partitioned as $M = \bigcup_{\ell\in[k]} M_\ell$, where each set $M_\ell$ consists of copies of the same chore.
\end{itemize}

An \textit{integral allocation} $\x = (\x_1, \x_2, \ldots, \x_n)$ is a partition of the chores into $n$ bundles, where agent $i$ receives bundle $\x_i \subseteq M$ and gets disutility $d_i(\x_i)$. In a \textit{fractional allocation} $\x \in [0, 1]^{n \times m}$, chores are divisible and $x_{ij} \in [0, 1]$ denotes the fraction of chore $j$ given to agent $i$. Here $d_i(\x_i) = \sum_{j \in M} d_i(j) \cdot x_{ij}$. We will assume that allocations are integral unless explicitly stated otherwise.

\paragraph{Fairness and efficiency notions.}
An allocation $\x$ satisfies:
\begin{enumerate}[leftmargin=*]
\item Envy-free up to one chore (EF1) for symmetric agents if for all $i,h\in N$, $d_i(\x_i \setminus j) \leq d_i(\x_h)$ for some $j\in \x_i$. 
\item Weighted envy-free up to one chore (wEF1) if for all $i,h\in N$, $\frac{d_i(\x_i \setminus j)}{w_i} \leq \frac{d_i(\x_h)}{w_h}$ for some $j\in \x_i$. 
\item Pareto-optimal if there is no allocation $\y$ that dominates $\x$. An allocation $\y$ dominates an allocation $\x$ if for all $i \in N$, $d_i(\y_i) \leq d_i(\x_i)$, and there exists $h \in N$ such that $d_h(\y_h) < d_h(\x_h)$.
\item Fractionally Pareto-optimal if there is no fractional allocation that dominates $\x$. An fPO allocation is clearly PO, but not vice-versa.
\end{enumerate}

\paragraph{Competitive equilibrium of chores.} In the Fisher market model for chores, we associate payments $\p = (p_1, \ldots, p_m)$ with the chores. Each agent $i$ aims to earn a desired \textit{minimum payment} of $e_i \geq 0$ by performing chores in exchange for payment. In a (fractional) allocation $\x$ with payments $\p$, the \textit{earning} of agent $i$ is $\p(\x_i) = \sum_{j \in M} p_j \cdot x_{ij}$. For each agent $i$, we define the \textit{pain-per-buck} ratio $\alpha_{ij}$ of chore $j$ as $\alpha_{ij} = d_i(j) / p_j$ and the \textit{minimum-pain-per-buck} (MPB) ratio as $\alpha_i = \min_{j \in M} \alpha_{ij}$. Further, we let $\mpb_i = \{j \in M \mid d_i(j) / p_j = \alpha_i\}$ denote the set of chores which are MPB for agent $i$ under payments $\p$. 

We say that $(\x, \p)$ is a \textit{competitive equilibrium} (CE) if (i) for all $j \in M$, $\sum_{i \in N} x_{ij} = 1$, i.e., all chores are completely allocated, (ii) for all $i \in N$, $\p(\x_i) = e_i$, i.e., each agent receives her minimum payment, and (iii) for all $i \in N$, $\x_i \subseteq \mpb_i$, i.e., agents receive only chores which are MPB for them. The First Welfare Theorem \cite{mas1995microeconomic} shows that competitive equilibria are efficient, i.e., for a CE $(\x, \p)$, the allocation $\x$ is fPO.

For a CE $(\x, \p)$ where $\x$ is integral, we let $\p_{-1}(\x_i) := \min_{j \in \x_i} \p(\x_i \setminus j)$ denote the payment agent $i$ receives from $\x_i$ excluding her highest paying chore. 

\begin{definition}(Weighted payment EF1)\label{def:wpef1}
An allocation $(\x, \p)$ is said to be \textit{weighted payment envy-free up to one chore} (wpEF1) if for all $i, h \in N$ we have $\frac{\p_{-1}(\x_i)}{w_i} \leq \frac{\p(\x_h)}{w_h}$. Agent $i$ \textit{wpEF1-envies} $h$ if $\frac{\p_{-1}(\x_i)}{w_i} > \frac{\p(\x_h)}{w_h}$.  
\end{definition}

The following lemma shows a sufficient condition for computing a wEF1 and PO allocation.
\begin{lemma}\label{lem:pEF1impliesEF1}
Let $(\x,\p)$ be a CE with $\x$ integral. If $(\x,\p)$ is wpEF1, then $\x$ is wEF1 and fPO.
\end{lemma}
\begin{proof}
Let $\alpha_i$ be the MPB ratio of agent $i$ in $(\x,\p)$. Consider any pair of agents $i,h\in N$. We have:
\begin{equation*}
\min_{j\in\x_i}\frac{ d_i(\x_i\setminus j)}{w_i} = \frac{\alpha_i\cdot\p_{-1}(\x_i)}{w_i} \le \frac{\alpha_i\cdot \p(\x_h)}{w_h} \le \frac{d_i(\x_h)}{w_h},
\end{equation*}
where the first and last transitions use that $(\x,\p)$ is on MPB, and the middle inequality uses the \cref{def:wpef1}. This shows that $\x$ is wEF1. Moreover, the First Welfare Theorem \cite{mas1995microeconomic} implies that the allocation $\x$ is fPO since $(\x, \p)$ is a CE.
\end{proof}

An agent $\ell$ is called a \textit{weighted least earner} (wLE) among agent set $A$ if $\ell \in \s{argmin}_{i \in A} \frac{\p(\x_i)}{w_i}$. An agent $b$ is called a \textit{weighted big earner} (wBE) among agent set $A$ if $b \in \s{argmax}_{i \in A} \frac{\p_{-1}(\x_i)}{w_i}$. The next lemma shows the importance of the wBE and wLE agents.

\begin{lemma}\label{lem:BEenvyLE}
An integral CE $(\x,\p)$ is wpEF1 if and only if a wBE $b$ does not wpEF1-envy a wLE $\ell$.
\end{lemma}
\begin{proof}
$(\Rightarrow)$ If $(\x, \p)$ is wpEF1 then clearly $b$ does not wpEF1-envy $\ell$. 

\noindent $(\Leftarrow)$ Suppose $b$ does not wpEF1-envy $\ell$. Then the following shows that no agent $i$ wpEF1-envies any agent $h$, implying that $(\x,\p)$ must be wpEF1.
\[ 
\frac{\p_{-1}(\x_i)}{w_i} \le \frac{\p_{-1}(\x_b)}{w_b} \le \frac{\p(\x_\ell)}{w_\ell} \le \frac{\p(\x_h)}{w_h}, 
\]
where the first and last inequalities use the definitions of wBE and wLE, and the middle inequality uses that $b$ does not wpEF1-envy $\ell$.
\end{proof}

\section{wEF1 and fPO Allocations in Three-Agent-Types Instances}
We now study the existence and computation of wEF1 and fPO allocations in instance with three agent types. We devise a polynomial time algorithm, \cref{alg:three-agent-types}, and show that it computes a wEF1 and fPO allocation for such instances. We therefore have the following theorem.
\begin{theorem}\label{thm:three-agent-types}
In any three-agent-types instance, a wEF1 and fPO allocation exists, and can be computed in polynomial time.
\end{theorem}
We remark that our result generalizes the following previously known results regarding the polynomial time computability of allocations that are: (i) EF1+fPO for three unweighted agents \cite{GMQ23chores}, (ii) EF1+fPO for two types of unweighted agents \cite{GMQ23chores}, (iii) wEF1+fPO for two agents \cite{wu2023wef1}.

Let $N = N_1 \sqcup N_2 \sqcup N_3$ be a partition of the set of agents into three sets, called agent \textit{groups}, each containing agents of the same type. Let $d_i(\cdot)$ be the disutility function of agents in group $N_i$, for $i\in\{1,2,3\}$. Note that agents in the same group can have different entitlements.

\begin{algorithm}[t]
\caption{Weighted Picking Sequence ($\wps$) Algorithm}\label{alg:wps}
\textbf{Input:} Agents $N$ with identical disutility function $d(\cdot)$, set of chores $M$ s.t. $d_1\ge \dots \ge d_m$\\
\textbf{Output:} An wEF1 allocation $\x$
\begin{algorithmic}[1]
\State For each $i\in N$, $\x_i \gets \emptyset$, $s_i \gets 0$ 
\For{$j=1$ to $m$}
\State $i \gets \arg\min_{k\in N} \frac{s_k}{w_k}$; ties broken in favour of smaller index agent
\State $\x_i \gets \x_i \cup \{j\}$, $s_i\gets s_i+1$
\EndFor
\State \Return $\x$
\end{algorithmic}
\end{algorithm}

\subsection{Algorithm Description} 
Through the execution of \cref{alg:three-agent-types} we let $M_i$ denote the set of chores allocated to group $N_i$. Initially, all chores are allocated to a single group, $N_1$, i.e., $M_1 = M$, and $M_2 = M_3 = \emptyset$. At each point, the set of chores $M_i$ is allocated to the group $N_i$ by using the \textit{weighted picking sequence} (WPS) procedure, \cref{alg:wps}, described below. 

\paragraph{Weighted Picking Sequence (WPS) Algorithm.} The WPS algorithm (\cref{alg:wps}) takes as input a set of agents $N$ of the same type, i.e., with \textit{identical} disutility function $d(\cdot)$, and a set of chores $M$. The algorithm first sorts and re-labels the $m$ chores in non-increasing order of disutility, i.e., $d_1 \ge d_2 \ge \dots \ge d_m$. The algorithm then performs $m$ iterations, allocating one chore in each iteration. In iteration $j$, chore $j$ is allocated to an agent with the least value of $\frac{s_i}{w_i}$, where $s_i$ denotes the number of chores allocated to agent $i$ so far. It is known that the WPS procedure returns a wEF1 allocation for identical agents \cite{chak2020wef1,wu2023wef1}. Recently, \cite{wu2023wef1} showed that by allocating the chores in reverse order with each agent picking their least disutility chore in their turn results in a wEF1 allocation, even for non-identical agents. Since our algorithm uses the WPS algorithm to assign chores to agents of the same type, we do not require this generalization. Using the WPS algorithm to allocate chores $M_i$ to group $N_i$ ensures that agents within a group do not ever wEF1-envy each other, and any wEF1-envy is between agents belonging to different groups. To reduce wEF1-envy across groups, we transfer chores from one group to another. After each chore transfer the set of chores $M_i$ gets updated and is re-allocated to the group $N_i$ by using the WPS algorithm. We denote by $\wps(N_1, M_1) \cup \wps(N_2, M_2) \cup \wps(N_3, M_3)$ the allocation resulting from allocating $M_i$ to $N_i$ using WPS for each $i\in[3]$. 

\paragraph{Chore transfers and Payment drops: High-level Ideas.} To ensure that the resulting allocation is fPO, we attach payments to the chores and maintain that all allocations in the run of \cref{alg:three-agent-types} are on MPB, i.e., are competitive. \cref{alg:three-agent-types} performs two kinds of \textit{steps}: (i) \textit{chore transfers} and (ii) \textit{payment drops}. Chore transfers involve the transfer of a chore from one group to another, and payment drops involve decreasing the payments of all chores belonging to one or two agent groups.

Initially we set the payment of chore $j$ as $p_j = d_1(j)$. At some point in the algorithm let the \textit{weighted big earner group} $N_\beta$ be the group containing the weighted big earner (wBE) at the time. Likewise, let the \textit{weighted least earner group} $N_\lambda$ be the group containing the weighted least earner (wLE) at the time. We will call the group $N_\mu$ which contains neither as the \textit{weighted middle earner (wME) group}. If $\beta = \lambda$ then the allocation must already be wpEF1 by \cref{lem:BEenvyLE}. Initially $N_\beta = N_1$. We maintain that for the majority of the algorithm $N_\beta = N_1$, i.e., the wBE is in $N_1$, and transfer chores unilaterally from $N_1$ to $N_2$ and $N_3$. We also aim to maintain that agents in $N_2$ and $N_3$ do not wpEF1-envy each other, and perform chore transfers between $N_2$ and $N_3$ to eliminate any arising wpEF1-envy. 

\begin{algorithm}[!t]
\caption{wEF1+fPO for three-agent-types instances}\label{alg:three-agent-types}
\textbf{Input:} Three-agent-types instance $(N,W,M,D)$ \\
\textbf{Output:} A wEF1 and fPO integral allocation $\x$
\begin{algorithmic}[1]
\State Let $N = N_1 \sqcup N_2 \sqcup N_3$ be the partition of the agents according to their type
\State Let $d_i(\cdot)$ denote the disutility function of agents in group $N_i$, for $i\in\{1,2,3\}$
\State $M_1 \gets M, M_2 \gets \emptyset, M_3 \gets \emptyset$
\State $\x \gets \wps(N_1, M_1) \cup \wps(N_2, M_2) \cup \wps(N_3, M_3)$
\State For each $j \in M$, set $p_j \gets d_1(j)$
\While{$\x$ is not wpEF1}
    \State $b \gets \s{argmax}_{k \in N} \p_{-1}(\x_k)/w_k$ \Comment{Weighted big earner}
    \State $\ell \gets \s{argmin}_{k \in N} \p(\x_k)/w_k$\Comment{Weighted least earner}
    \State $\beta \gets i \in [3]$ s.t. $b \in N_i$, 
    \Comment{Index of wBE group}
    \State $\lambda \gets i \in [3]$ s.t. $\ell \in N_i$ \Comment{Index of wLE group}
    \State $\mu \gets i \in [3] \setminus \{\beta, \lambda\}$ \Comment{Index of wME group}
    \Statex \quad\; {$\triangleright$ Check for $N_\beta$ to $N_\lambda$ chore transfer} 
    \If{$\exists j \in M_\beta \cap \mpb_{\lambda}$} 
        \State $M_\beta \gets M_\beta \setminus j$, $M_{\lambda} \gets M_{\lambda} \cup j$
        \State $\x \gets \bigcup_{i\in[3]}\wps(N_i, M_i)$
    \Statex \quad\; {$\triangleright$ Check for potential $N_\mu$ to $N_\lambda$ chore transfer}
    \ElsIf{$|M_\mu \cap \mpb_{\lambda}| > 0$} 
        \Statex \quad\quad\quad $\triangleright$ Check if $N_\mu$ to $N_\lambda$ chore transfer is allowed
        \If{$\exists j\in M_\mu \cap \mpb_{\lambda}$ s.t. $\min\limits_{k \in N_\mu} \frac{\p(\y_k)}{w_k} > \frac{\p(\x_\ell)}{w_\ell}$ for $\y = \wps(N_\mu, M_\mu \setminus j)$}
            \State $M_\mu \gets M_\mu \setminus j$, $M_{\lambda} \gets M_{\lambda} \cup j$
            \State $\x \gets \bigcup_{i\in[3]}\wps(N_i, M_i)$
        \Statex \quad\quad\quad{$\triangleright$ $N_\beta$ to $N_\mu$ chore transfer}
        \ElsIf{$\exists j' \in M_\beta \cap \mpb_{\mu}$}
            \State $M_\beta \gets M_\beta \setminus j'$, $M_\mu \gets M_\mu \cup j'$
            \State $\x \gets \bigcup_{i\in[3]}\wps(N_i, M_i)$
        \Else \Comment{No chore can be transferred from $N_\beta$}
            \Statex \quad\quad\quad $\triangleright$ Lower payments of chores in $M_\mu \cup M_\lambda$
            \State $\gamma \gets \max_{i \in \{\lambda, \mu\}, j \in M_\beta} \frac{\alpha_i}{d_i(j)/p_{j}}$; $\alpha_i$ is the MPB ratio  of an agent in group $N_i$
            \For{$j \in M_\mu \cup M_\lambda$} 
                \State $p_j \gets \gamma\cdot p_j$ 
            \EndFor
        \EndIf
    \Else \Comment{No chore can be transferred from $N_\beta$ or $N_\mu$}
    \Statex \quad\quad\quad $\triangleright$ Lower payments of chores in $M_\lambda$
        \State $\gamma \gets \max_{j \in M_\beta \cup M_\mu} \frac{\alpha_{\ell}}{d_{\ell}(j)/p_j}$
            \For{$j \in M_\lambda$}
                \State $p_j \gets \gamma \cdot p_j$ 
            \EndFor        
    \EndIf
\EndWhile
\State \Return $\x$
\end{algorithmic}
\end{algorithm}

Since $N_1$ only loses chores, in at most $m$ chore transfers from $N_1$ it must be that either the allocation becomes wpEF1 (hence wEF1), or the wBE ceases to be in $N_1$. In the latter case we show that in one additional chore transfer step the allocation is wEF1. To ensure fPO and facilitate chore transfers, we perform payment drops of chores when appropriate, and maintain the MPB condition during chore transfers. That is, a chore $j$ is transferred from group $N_i$ to $N_h$ only if $j$ belongs to the MPB set $\mpb_h$ of agents in $N_h$, i.e., $j\in M_i \cap \mpb_h$.

\paragraph{Chore transfers and Payment drops: Details.} We perform payment drops and chore transfers across groups in a careful and specific manner, as described below. 
\begin{itemize}[leftmargin=*]
\item[](Lines 12-14) Since we desire a wpEF1 allocation if possible we first check if there exists a chore $j\in M_\beta \cap \mpb_\lambda$ which can be transferred directly from the wBE group to the wLE group. If so we make this transfer.
\item[](Lines 15-18) If no chore can be transferred from $N_\beta$ to $N_\lambda$, we check if a chore in $M_\mu \cap \mpb_\lambda$ can be potentially transferred from $N_\mu$ to $N_\lambda$. If there is a chore $j\in M_\mu \cap \mpb_\lambda$ s.t. after losing $j$ the least weighted earning of an agent in $N_\mu$ is strictly larger than the weighted earning of the current wLE, then we transfer $j$ from $N_\mu$ to $N_\lambda$ (Line 17-18). Line 16 performs this check. If not, an $N_\mu$ to $N_\lambda$ transfer is \textit{not allowed}.
\item[](Lines 19-21) If an $N_\mu$ to $N_\lambda$ transfer is not allowed, we check if an $N_\beta$ to $N_\mu$ transfer is possible, i.e., there is a chore $j'\in M_\beta \cap \mpb_\mu$. If so, we make such a transfer. 
\item[](Lines 22-25) Otherwise, no chore can be transferred from $N_\beta$ to either $N_\mu$ or $N_\lambda$. In this case we lower the payments of chores in $M_\mu$ and $M_\lambda$ until a chore in $M_\beta$ becomes MPB for agents in $N_\mu$ or $N_\lambda$.
\item[](Lines 26-29) Finally if there is no chore that can be transferred from $N_\beta$ or $N_\mu$ to $N_\lambda$, we lower the payments of chores in $M_\lambda$ until a chore in $M_\beta$ or $M_\mu$ becomes MPB for agents in $N_\lambda$.
\end{itemize}

\subsection{Analysis of \cref{alg:three-agent-types}: Overview}
We begin the analysis of \cref{alg:three-agent-types} by proving some important properties of the WPS algorithm (Alg. \ref{alg:wps}). The following lemmas compare the disutilities of agents before and after a chore transfer.

\begin{restatable}{lemma}{lemwpstransferone}\label{lem:wps-transfer-1}
Let $N$ be a set of agents with the same disutility function $\p(\cdot)$ and $M$ be a set of chores. Let $\x = \wps(N, M)$ be an allocation of $M$ to $N$ using the WPS algorithm, and $\y = \wps(N, M\setminus j)$, for a chore $j\in M$. Then we have:
\begin{enumerate}
\item[(i)] For all $i\in N$, $\p(\y_i) \le \p(\x_i)$.
\item[(ii)] For all $i\in N$, $\p_{-1}(\y_i) \le \p_{-1}(\x_i)$.
\item[(iii)] For all $i, h\in N$, $\frac{\p(\y_i)}{w_i} \ge \frac{\p_{-1}(\x_h)}{w_h}$. In particular, $\frac{\p(\y_\ell)}{w_\ell} \ge \frac{\p_{-1}(\x_b)}{w_b}$, where agent $\ell$ has the least weighted disutility in $\y$ and agent $b$ has the biggest weighted disutility up to one chore in $\x$.  
\end{enumerate}
\end{restatable}

\begin{restatable}{lemma}{lemwpstransfertwo}\label{lem:wps-transfer-2}
Let $N$ be a set of agents with the same disutility function $\p(\cdot)$ and $M$ be a set of chores. Let $\x = \wps(N, M)$ be an allocation of $M$ to $N$ using the WPS algorithm, and $\y = \wps(N, M\cup j)$, for a chore $j\notin M$. Then we have:
\begin{enumerate}
\item[(i)] For all $i\in N$, $\p(\y_i) \ge \p(\x_i)$. 
\item[(ii)] For all $i\in N$, $\p_{-1}(\y_i) \ge \p_{-1}(\x_i)$.
\item[(iii)] For all $i, h\in N$, $\frac{\p_{-1}(\y_h)}{w_h} \le \frac{\p(\x_i)}{w_i}$. In particular, $\frac{\p_{-1}(\y_b)}{w_b} \le \frac{\p(\x_\ell)}{w_\ell}$, where agent $b$ has the biggest weighted disutility up to one chore in $\y$ and agent $\ell$ has the least weighted disutility in $\x$.  
\end{enumerate}
\end{restatable}

When a set of chores $M$ with associated payments $\p$ is on MPB for an group $N$ containing agents of the same type, the disutility of a chore is proportional to its payment. Therefore, the claims in the above lemmas are also true when $\p$ represents the payment vector (which is proportional to the disutility vector) and $\p(\x_i)$ represents the earning of agent $i$ in $(\x,\p)$ (which is proportional to the disutility of $i$ in $\x$). We next show:

\begin{restatable}{lemma}{leearning}\label{lem:le-earning}
Throughout the execution of \cref{alg:three-agent-types}, no chore transfer decreases the weighted earning of the weighted least earner.
\end{restatable}

The following two lemmas analyze steps of \cref{alg:three-agent-types} which cause a group $N_\lambda$ to cease being the weighted least earner group. 
\begin{restatable}{lemma}{letobe}\label{lem:le-to-be}
If a step of \cref{alg:three-agent-types} results in the weighted least earner group becoming the weighted big earner group, then the resulting allocation must be wpEF1.
\end{restatable}

\begin{restatable}{lemma}{lechange}\label{lem:le-change}
If a step of \cref{alg:three-agent-types} results in the weighted least earner group $N_\lambda$ becoming the weighted middle earning group, then agents in $N_\lambda$ do not wpEF1-envy agents in the new weighted least earner group in the resulting allocation.
\end{restatable}

To summarize, if a group $N_\lambda$ ceases to be the weighted least earner group, then either resulting allocation is wpEF1 or $N_\lambda$ becomes the weighted middle earning group and agents in $N_\lambda$ do not wpEF1-envy agents in the new weighted least earner group. These observations are important for proving the lemmas that follow. 

We next consider steps of \cref{alg:three-agent-types} which cause a group $N_\beta$ cease being the weighted big earner group. We first show that:
\begin{restatable}{lemma}{bebecomesle}\label{lem:be-becomes-le}
If a step of \cref{alg:three-agent-types} results in the weighted big earner group becoming the weighted least earner group, then the resulting allocation must be wpEF1.
\end{restatable}

Recall that we initially assigned all the chores to group $N_1$. Hence $N_\beta = N_1$ was the initial weighted big earner group. We prove that this is the case almost throughout the execution of the algorithm. To this end, we show that $N_\beta$ (i.e. $N_1$) loses a chore in every $\poly{m}$-many steps. 
\begin{restatable}{lemma}{beloseschore}\label{lem:be-loses-chore}
While the weighted big earner belongs to the group $N_\beta$ and the allocation is not wpEF1, $N_\beta$ must lose a chore in $\poly{m}$ steps.
\end{restatable}

Note that $N_1$ has $m$ chores to begin with (i.e. always $|M_1|\le m$), and \cref{alg:three-agent-types} never transfers a chore to $N_1$. \cref{lem:be-loses-chore} therefore implies that in $\poly{m}$ steps either the allocation is wpEF1 or $N_1$ ceases to be the weighted big earner group. In the latter scenario, we show that we arrive at a wEF1 allocation in at most one more chore transfer step.

\begin{restatable}{lemma}{bechange}\label{lem:be-change}
After $N_1$ stops being the weighted big earner group for the last time, \cref{alg:three-agent-types} terminates with a wEF1 and fPO allocation after performing at most one subsequent chore transfer.
\end{restatable}

The above discussion leads us to conclude that \cref{alg:three-agent-types} computes a wEF1 and fPO allocation in polynomial time for three-agent-types instances. This proves \cref{thm:three-agent-types}.

\subsection{Analysis of \cref{alg:three-agent-types}: Proofs}
This section proves the lemmas stated in the previous section.

\lemwpstransferone*
\begin{proof}
We use the function $\sigma:[m]\rightarrow[n]$ to represent the $m$-length picking sequence used by the WPS algorithm on an instance $(N,M,\p)$. Thus if $\sigma(c) = i$ for a chore $c$ and agent $i$, then in iteration $c$ of the WPS algorithm (Alg.~\ref{alg:wps}) agent $i\in\arg\min_k s_k/w_k$ and agent $i$ is assigned chore $c$. We now consider the allocation $\y$ produced by the WPS algorithm on the instance $M\setminus j$ for some chore $j\in [m]$. For conceptual convenience, we also add a dummy item $(m+1)$ with $p_{m+1} = 0$. This ensures that for all $i$, $|\x_i| = |\y_i|$. The addition of the dummy item has no effect on either the values of disutility or the disutility minus the worst chore for any agent. 
Moreover the dummy item will necessarily be the last item to be allocated in the picking sequence $\sigma':[m+1]\setminus \{j\}\rightarrow[n]$ generating $\y$. They key observation is that since the WPS algorithm only depends on the number of items $s_i$ allocated to an agent $i$, the sequence $\sigma'$ is the same as $\sigma$ up to item $(j-1)$, after which it is `left-shifted'. That is, for $j'\in [m+1]\setminus\{j\}$, we have:
\begin{equation}\label{eq:wps-shift}
\sigma'(j') = 
\begin{cases}
    \sigma(j'), \text{ if } j' < j, \\
    \sigma(j'-1), \text{ if } j' > j.
\end{cases}    
\end{equation}

For an agent $i$, let $\x_i = \{a_1,\dots,a_k\}$ be the items in $M = [m]$ assigned to $i$ in $\x$ in the order in which they are picked by $i$. Thus, $p_{a_1} \ge p_{a_2} \ge \dots \ge p_{a_k}$. Let $\y_i = \{b_1,\dots,b_k\}$ be the items assigned to $i$ in $\y$ in the order in which they are picked by $i$. \cref{eq:wps-shift} implies that for $t\in[k]$:
\begin{equation}\label{eq:wps-shift-2}
b_t = 
\begin{cases}
    a_t, \text{ if } a_t < j \\
    a_t + 1, \text{ if } a_t \ge j
\end{cases}        
\end{equation}
Since $p_{a_t}\ge p_{a_t+1}$, \cref{eq:wps-shift-2} immediately implies that $p_{b_t} \le p_{a_t}$ for all $t\in[k]$. Thus we have: 
\[
\p(\y_i) = \sum_{t=1}^k p_{b_t} \le \sum_{t=1}^k p_{a_t} = \p(\x_i), \text{ and }
\p_{-1}(\y_i) = \sum_{t=2}^k p_{b_t} \le \sum_{t=2}^k p_{a_t} = \p_{-1}(\x_i), 
\]
which proves parts (i) and (ii).

To prove (iii), we use the continuous interpretation of the WPS algorithm \cite{wu2023wef1}. We imagine $s_i:[0,m]\rightarrow [0,m]$ as a non-decreasing function that grows uniformly at a rate of $1/w_i$ during time $(j-1, j]$ if $i$ was chosen in iteration $j$, and is unchanged during other time intervals. Define a function $\varphi: (0, k/w_i] \rightarrow \R^+$ such that $\varphi(\alpha)$ is the disutility of the item being picked by agent $i$ when $s_i(t) = \alpha$ in allocation $\x$. In other words, for $z\in[k]$ and $\alpha\in \big(\frac{z-1}{w_i},\frac{z}{w_i}\big]$, define $\varphi(\alpha) = p_{a_z}$. Similarly define $\varphi':(0, k/w_i] \rightarrow \R^+$ such that $\varphi'(\alpha)$ is the disutility of the item being picked by agent $i$ when $s_i(t) = \alpha$ in allocation $\y$. In other words, for $z\in[k]$ and $\alpha\in \big(\frac{z-1}{w_i},\frac{z}{w_i}\big]$, define $\varphi'(\alpha) = p_{b_z}$. These definitions imply that:
\begin{equation}\label{eq:y-earning}
\frac{\p(\y_i)}{w_i} = \int_0^{\frac{k}{w_i}} \varphi'(\alpha) \: d\alpha.
\end{equation}

Consider another agent $h$ and let $\x_h = \{e_1,\dots,e_\ell\}$. As before, define $\phi: (0, \ell/w_h] \rightarrow \R^+$ such that $\phi(\alpha)$ is the disutility of the item being picked by agent $h$ when $s_h(t) = \alpha$ in allocation $\x$. In other words, for $z\in[\ell]$ and $\alpha\in \big(\frac{z-1}{w_h},\frac{z}{w_h}\big]$, define $\phi(\alpha) = p_{e_z}$. This implies:
\begin{equation}\label{eq:x-earning}
\frac{\p_{-1}(\x_h)}{w_h} = \int_{\frac{1}{w_h}}^{\frac{\ell}{w_h}} \phi(\alpha) \: d\alpha.
\end{equation}

We now prove the following claims.
\begin{claim}\label{clm:wps-one}
$\frac{\ell-1}{w_h} \le \frac{k}{w_i}$.
\end{claim}
\begin{proof}
At the time $t$ when agent $h$ picks her last item $e_\ell$, it must be that $s_h(t) \le s_i(t)$, since $h$ was chosen at time $t$. The claim then follows by noting $s_h(t) = (\ell-1)/w_h$, and $s_i(t) \le s_i(m) = k/w_i$ since $s_i$ is non-decreasing. 
\end{proof}

\begin{claim}\label{clm:wps-two}
For all $\alpha\in\big(\frac{1}{w_h}, \frac{\ell}{w_h}\big]$, $\phi(\alpha) \le \varphi'(\alpha-\frac{1}{w_h})$.
\end{claim}
\begin{proof}
Let $z\in [\ell]$ be such that $\alpha\in\big( \frac{z-1}{w_h}, \frac{z}{w_h}\big]$. Then $\phi(\alpha) = p_{e_z}$ by definition of $\phi$. Let $t_2$ be the time instant s.t. $s_h(t_2) = \alpha$ for the first time. Let $t^*$ be the largest integer strictly smaller than $t_2$, i.e., $t_2\in (t^*, t^*+1]$. In this interval $s_h(t)$ goes from $\frac{z-1}{w_h}$ to $\frac{z}{w_h}$.  Now let $t_1$ be the earliest time when $s_i(t_1) = \alpha-\frac{1}{w_h}$. Note that $t_1$ is well defined, since $\alpha-\frac{1}{w_h} > 0$, as $\alpha \in (\frac{1}{w_h}, \frac{\ell}{w_h}]$. Thus we have $s_i(t_1) = \alpha - \frac{1}{w_h} \le \frac{z-1}{w_h} = s_h(t^*)$. At $t^*$, since agent $h$ is the agent chosen to pick an chore, hence it must be that $s_h(t^*) \le s_i(t^*)$. Therefore $s_i(t_1) \le s_i(t^*)$. Since $s_i(\cdot)$ is non-decreasing, we conclude $t_1\le t^*$. Since $t_2 \in (t^*, t^*+1]$, we have $t_1 < t_2$.

At time $t_1$, let $a$ (resp. $b$ be the chore being consumed by agent $i$ in allocation $\x$ (resp. $\y)$. From \cref{eq:wps-shift-2}, we know that either $b = a$ or $b = a+1$. Recall that $e_z$ is the chore being consumed by agent $h$ at time $t_2$ in allocation $\x$. Since $t_1 < t_2$, and since chores are allocated in non-increasing order of disutility, it must be that the disutility of $e_z$ does not exceed the disutility of chore $(a+1)$, since $(a+1)$ is the chore with the highest disutility available after chore $a$ gets consumed. Thus $p_{e_z} \le p_{a+1}$ and hence $p_{e_z} \le p_b$. The claim then follows by noting that $\phi(\alpha) = p_{e_z}$ and $p_b = \varphi'(\alpha-\frac{1}{w_h})$. The latter equality holds from the definition of $\varphi'$, as $b$ is the chore being consumed by $h$ at the time $t_1$ when $s_h(t_1) = \alpha-\frac{1}{w_h}$.
\end{proof}

Armed with Claims~\ref{clm:wps-one} and \ref{clm:wps-two}, we observe that:

\begin{equation*}
\begin{aligned}
\frac{\p(\y_i)}{w_i} &= \int_0^{\frac{k}{w_i}} \varphi'(\alpha) \: d\alpha &&\text{ (\cref{eq:y-earning}) }\\
&\ge \int_0^{\frac{\ell-1}{w_h}} \varphi'(\alpha) \: d\alpha &&\text{ (\cref{clm:wps-one}) }\\
&= \int_{\frac{1}{w_\ell}}^{\frac{\ell}{w_h}} \varphi'(\alpha+1/w_h) \: d\alpha &&\text{ (change of variables) }\\
&\ge \int_{\frac{1}{w_\ell}}^{\frac{\ell}{w_h}} \phi(\alpha) \: d\alpha &&\text{ (\cref{clm:wps-two}) }\\
&= \frac{\p_{-1}(\x_h)}{w_h}, &&\text{ (\cref{eq:x-earning}) } 
\end{aligned} 
\end{equation*} 
thus proving part (iii) of the lemma.
\end{proof}

\lemwpstransfertwo*
\begin{proof}
Define $M' = M\cup j$. Then $\y = \wps(N, M')$ and $\x = \wps(N, M'\setminus j)$. \cref{lem:wps-transfer-2} then directly follows by invoking \cref{lem:wps-transfer-1} on allocations $\y$ and $\x$ (with labels $\x$ and $\y$ swapped).    
\end{proof}

\leearning*
\begin{proof}
Consider a transfer of a chore which strictly reduces the weighted earning of the weighted least earner (wLE). Let $\x$ (resp. $\y$) be the allocation before (resp. after) the transfer and $\p$ the payment vector. Note that $(\x,\p)$ is not wpEF1, otherwise \cref{alg:three-agent-types} would have halted before the transfer. 

\begin{claim}\label{clm:new-wle-group}
The new wLE in $(\y,\p)$ belongs to the group losing the chore.
\end{claim}
\begin{proof}
By \cref{lem:wps-transfer-1} (i), the new wLE cannot be in group receiving the chore. Moreover the new wLE cannot be in the group which does not participate in the transfer, since the earning of agents in that group is unchanged. Therefore the new wLE must belong to the group losing the chore.    
\end{proof}

Let $N_\beta$, $N_\mu$, and $N_\lambda$ be the wBE, wME, and wLE groups in $(\x, \p)$. There are two possibilities: 
\begin{itemize}
\item Chore transfer from $N_\beta$ to $N_\lambda$ or $N_\mu$ . Let $\ell$ and $b$ be the wLE and wBE in $(\x,\p)$ respectively, and $\ell'$ the new wLE in $(\y,\p)$. By definition we have $\ell\in N_\lambda$, $b\in N_\beta$, and by \cref{clm:new-wle-group}, $\ell'\in N_\beta$. Observe that:
\[ 
\frac{\p(\y_{\ell'})}{w_{\ell'}} \ge \frac{\p_{-1}(\x_b)}{w_b} > \frac{\p(\x_\ell)}{w_\ell},
\]
where the first inequality is due to \cref{lem:wps-transfer-1} applied to $N_\beta$ and the second inequality is because $(\x,\p)$ is not wpEF1.
\item Chore transfer from $N_\mu$ to $N_\lambda$. \cref{clm:new-wle-group} implies the new wLE in the allocation $(\y,\p)$ must be in $N_\mu$. By the design of \cref{alg:three-agent-types} a chore transfer from $N_\mu$ to $N_\lambda$ is allowed only if after the transfer the weighted earning of the wLE in $N_\mu$ is strictly larger than the weighted earning of the current wLE (Line 16). 
\end{itemize}
In either case the weighted earning of the wLE does not decrease due to the chore transfer.
\end{proof}

\letobe*
\begin{proof}
Clearly a payment drop step cannot cause the wLE group to change. Let us therefore consider a chore transfer step which causes the wLE group $N_\lambda$ in an allocation $(\x,\p)$ to become the wBE group in the subsequent allocation $(\y,\p)$. If a chore is not transferred to $N_\lambda$, the wLE in $(\y,\p)$ must continue to lie in $N_\lambda$. Then $(\y,\p)$ must be wpEF1 since both the wBE and wLE lie in the same group in $(\y,\p)$. Therefore suppose a chore is transferred to $N_\lambda$. Let $\ell$ be the wLE in $(\x,\p)$. Let $\ell'$ and $b'$ be the wLE and wBE in $(\y, \p)$ respectively. We know $\ell,b'\in N_\lambda$. Observe that:
\[ 
\frac{\p_{-1}(\y_{b'})}{w_{b'}} \le \frac{\p(\x_\ell)}{w_\ell} \le \frac{\p(\x_{\ell'})}{w_{\ell'}},
\]
where the first inequality uses \cref{lem:wps-transfer-2} applied to $N_\lambda$ and the second inequality uses \cref{lem:le-earning}. The above equation implies that the wBE does not wpEF1-envy the wLE in $(\y,\p)$, implying that $(\y,\p)$ is wpEF1 by invoking \cref{lem:BEenvyLE}.
\end{proof}

\lechange*
\begin{proof}
Clearly a payment drop step cannot cause the wLE group to change. Let us therefore consider a chore transfer step which causes the wLE group $N_\lambda$ in an allocation $(\x,\p)$ to become the wME group in the subsequent allocation $(\y,\p)$. If a chore is not transferred to $N_\lambda$, the wLE in $(\y,\p)$ must continue to lie in $N_\lambda$ and $N_\lambda$ cannot be the wME group in $(\y,\p)$. Therefore suppose a chore is transferred to $N_\lambda$. Let $\ell\in N_\lambda$ be the wLE in $(\x,\p)$, $b'$ be the wBE in $N_\lambda$ in $(\y,\p)$, and $\ell'$ the wLE in $(\y, \p)$. Observe that:
\[ 
\frac{\p_{-1}(\y_{b'})}{w_{b'}} \le \frac{\p(\x_\ell)}{w_\ell} \le \frac{\p(\x_{\ell'})}{w_{\ell'}},
\]
where the first inequality uses \cref{lem:wps-transfer-2} applied to $N_\lambda$ and the second inequality uses \cref{lem:le-earning}. The above equation implies that the wBE in $N_\lambda$ does not wpEF1-envy the wLE in $(\y,\p)$, implying that agents in $N_\lambda$ (the new wME group) do not wpEF1-envy agents in the new wLE group.
\end{proof}

\bebecomesle*
\begin{proof}
Since $N_1$ does not participate in any payment drops, only a chore transfer step can cause $N_1$ to stop being the wBE group. Let $j'\in M_1$ be the chore transferred from $N_1$ in the allocation $(\x,\p)$ which results in the allocation $(\x',\p)$ in which $N_1$ is the wLE group. We show $(\x',\p)$ is wpEF1. We have two possibilities for the transfer of $j'$.

\paragraph{Case 1.} Suppose $j'$ is transferred from $N_1$ to $N_\mu$ in $(\x,\p)$. We argue that $N_1$ cannot be the wLE group in $(\x',\p)$. For sake of contradiction assume this is the case. Let $\ell\in N_\lambda$ and $b\in N_1$ be the wLE and wBE in $(\x,\p)$ respectively, and let $\ell'\in N_1$ be the wLE agent in group $N_1$ in $(\x', \p)$. Observe that:
\[ 
\frac{\p(\x'_{\ell'})}{w_{\ell'}} \ge \frac{\p_{-1}(\x_b)}{w_b} > \frac{\p(\x_\ell)}{w_\ell} = \frac{\p(\x'_\ell)}{w_\ell},
\]
where the first inequality uses \cref{lem:wps-transfer-1} applied to group $N_1$, the second inequality uses the fact that $(\x,\p)$ is not wpEF1, and the equality holds because $\x_\ell = \x'_\ell$. This inequality implies that in $(\x',\p)$, the weighted earning of the wLE agent $\ell'$ in $N_1$ exceeds the weighted earning of the wLE agent $\ell$ in $N_\lambda$, which contradicts the assumption that $N_1$ is the wLE group in $(\x',\p)$.

\paragraph{Case 2.} Suppose $j'$ is transferred from $N_1$ to $N_\lambda$ in $(\x,\p)$. If $N_\lambda$ is the wBE group in $(\x',\p)$ then \cref{lem:le-to-be} implies $(\x',\p)$ is wpEF1. Let us therefore assume $N_\lambda$ is the wME group and $N_\mu$ is the wBE group in $(\x',\p)$. Let $b\in N_1$ be the wBE agent in $(\x,\p)$, and let $\ell'\in N_1$ and $b'\in N_\mu$ be the wLE and wBE agents in $(\x', \p)$ respectively. Observe that:
\[ 
\frac{\p_{-1}(\x'_{b'})}{w_{b'}} = \frac{\p_{-1}(\x_{b'})}{w_{b'}} \le \frac{\p_{-1}(\x_b)}{w_b} \le \frac{\p(\x'_{\ell'})}{w_{\ell'}},
\]
where the equality uses $\x'_i = \x_i$ for all $i\in N_\mu$, the first inequality uses the fact that $b$ is the wBE agent in $\x$, and the second inequality uses \cref{lem:wps-transfer-1} applied to group $N_1$. This shows that if $N_1$ is the wLE group in $(\x',\p)$, then it is wpEF1.
\end{proof}

\beloseschore*
\begin{proof}
We first note that there can be at most two payment drops (Lines 23-25 and Lines 27-29) after which a chore transfer from the wBE group $N_\beta$ must occur. We therefore only need to bound the number of chore transfers taking place between the other two groups. Clearly there can only be $m$ chore transfers from the wME group $N_\mu$ to the wLE group $N_\lambda$, after which $N_\lambda$ stops being the wLE group. Consider the last transfer of a chore $j\in M_\mu \cap \mpb_\lambda$ in an allocation $(\x,\p)$ which results in an allocation $\x'$ and causes $N_\lambda$ to stop being the wLE group. Using Lemmas \ref{lem:le-to-be} and \ref{lem:le-change}, if $(\x', \p)$ is not wpEF1 then in $(\x', \p)$, $N_\lambda$ is the new wME group and $N_\mu$ is the new wLE group and no agent in $N_\lambda$ wpEF1-envies an agent in $N_\mu$.

We claim that after the last transfer of the chore $j$ from $N_\mu$ to $N_\lambda$, \cref{alg:three-agent-types} will not perform any chore transfer from $N_\lambda$ to $N_\mu$ in $(\x', \p)$. Recall the \cref{alg:three-agent-types} allows transfers from the wME group to the wLE group only under certain conditions (Lines 16-18). We argue that these conditions are not met at $(\x',\p)$ and hence a transfer from $N_\lambda$ to $N_\mu$ cannot take place. In $(\x',\p)$, we have $j\in M_\lambda \cap \mpb_\mu$. Consider the allocation $\y = \wps(N_\lambda, \big(\bigcup_{i\in N_\lambda} \x'_i\big)\setminus j)$ as computed in Line 16 when deciding whether or not to allow a transfer from $N_\lambda$ to $N_\mu$ in $(\x',\p)$. In the allocation $\x'$, the group $N_\lambda$ is assigned all chores allocated to $N_\lambda$ including the transferred chore $j$, i.e., $\big(\bigcup_{i\in N_\lambda} \x'_i\big)\setminus j) = \big(\bigcup_{i\in N_\lambda} \x_i\big))$. Thus the allocation $\y$ is exactly the allocation $\x$ restricted to the group $N_\lambda$, i.e., for all $i\in N_\lambda$, $\y_i = \x_i$. We therefore have:
\[ 
\min_{i\in N_\lambda} \frac{\p(\y_i)}{w_i} = \min_{i\in N_\lambda} \frac{\p(\x_i)}{w_i} \le \min_{k\in N_\mu} \frac{\p(\x'_k)}{w_k},
\]
where the second inequality is due to \cref{lem:le-earning} as the wLE in $(\x',\p)$ lies in $N_\mu$. Therefore, \cref{alg:three-agent-types} will not proceed with a transfer from the wME group $N_\lambda$ to the wLE group $N_\mu$ in $(\x',\p)$. 

The next step of the algorithm will thus either be a chore transfer from $N_\beta$ or a payment drop of chores assigned to $N_\lambda \cup N_\mu$ immediately after which a chore is transferred from $N_\beta$. Therefore in every $\poly{m}$ steps a chore is taken from $N_\beta$.
\end{proof}

\bechange*
\begin{proof}
Since $N_1$ does not participate in any payment drops, only a chore transfer step can cause $N_1$ to stop being the wBE group. Let $(\x,\p)$ be the last allocation in the execution of \cref{alg:three-agent-types} in which $N_1$ is the wBE group. Let us suppose the transfer of a chore $j'\in M_1$ in $(\x,\p)$ results in the allocation $\x'$ where $N_1$ is no longer the wBE group. If $N_1$ is the wLE group in $(\x',\p)$ then \cref{lem:be-becomes-le} shows that $(\x',\p)$ is wpEF1. If $N_\lambda$ is the wBE group in $(\x,\p)$ then \cref{lem:le-to-be} implies $(\x',\p)$ is wpEF1. Therefore for the rest of the proof, we assume that

\begin{quote}
In $(\x',\p)$, $N_\lambda$ is the wLE group, $N_1$ is the wME group, and $N_\mu$ is the wBE group.
\end{quote}

\noindent Let $\ell\in N_\lambda$ and $b\in N_1$ be the wLE and wBE agents in $(\x,\p)$, and let $\ell'\in N_\lambda$ and $b'\in N_\mu$ be the wLE and wBE agents in $(\x',\p)$. We have two possibilities for the transfer of $j'$ in $(\x,\p)$.

\paragraph{Case 1.} Suppose $j'$ is transferred from $N_1$ to $N_\mu$ in $(\x,\p)$. Let $M_\mu$ and $M'_\mu$ be the set of chores allocated to agents in $N_\mu$ in allocations $\x$ and $\x'$ respectively. That is, $M_\mu = \bigcup_{i\in N_\mu} \x_i$, and $M'_\mu = \bigcup_{i\in N_\mu} \x'_i = M_\mu \cup j'$. The fact that \cref{alg:three-agent-types} performed the transfer of $j'$ from $N_1$ to $N_\mu$ implies that $|M_\mu \cap \mpb_\lambda| \neq\emptyset$, and for all $j\in M_\mu \cap \mpb_\lambda$ such that for $\y  = \wps(N_\mu, M_\mu\setminus j)$ as computed in Line 16, we have:
\begin{equation}\label{eq:case1}
\min_{k\in N_\mu} \frac{\p(\y_k)}{w_k} \le \frac{\p(\x_\ell)}{w_\ell}.
\end{equation}

Suppose $(\x',\p)$ is not wpEF1. Notice that in $(\x',\p)$, there is a chore $j$ which belongs to the wBE group $N_\mu$ and lies in the MPB set of the wLE group $N_\lambda$, i.e., $j\in M'_\mu \cap \mpb_\lambda$. This is because we know there exists some chore $j\in M_\mu \cap \mpb_\lambda$ and $M'_\mu = M_\mu \cup j'$. Thus as per Lines 12-14, \cref{alg:three-agent-types} will transfer chore $j$ from the wBE group $N_\mu$ to the wLE group $N_\lambda$ of $(\x',\p)$, resulting in an allocation $\x''$. We show that $(\x'',\p)$ is wpEF1 by considering the identity of the wBE group in $(\x'',\p)$ as follows.

\begin{itemize}
\item[(i)] In $(\x'',\p)$, $N_\lambda$ is the wBE group. This means that a chore transfer step caused the wLE group $N_\lambda$ of $(\x',\p)$ to become the wBE group in $(\x'',\p)$. By \cref{lem:le-to-be} the allocation $(\x'',\p)$ must be wpEF1.

\item[(ii)] In $(\x'',\p)$, $N_1$ is the wBE group. This contradicts the fact that $(\x,\p)$ was the last allocation in which $N_1$ is the wBE group.

\item[(iii)] In $(\x'',\p)$, $N_\mu$ is the wBE group. 
Let $\ell''$ be the wLE agent and $b''\in N_\mu$ be the wBE agent in $(\x'',\p)$. 

Let $\hat{\x}''$ be the allocation $\x''$ restricted to the agents in group $N_\mu$. Then we have $M''_\mu := \bigcup_{i\in N_\mu} \hat{\x}''_i = M'_\mu \setminus j = (M_\mu \setminus j) \cup j'$. Note that $\hat{\x}'' = \wps(N_\mu, M''_\mu)  = \wps(N_\mu, (M_\mu \setminus j) \cup j')$, and $\y = \wps(N_\mu, M_\mu \setminus j)$. Hence we can apply \cref{lem:wps-transfer-2} to the group $N_\mu$ and the allocations $\hat{\x}''$ and $\y$ to obtain: 
\[ 
\frac{\p_{-1}(\x''_{b''})}{w_{b''}} \le \min_{k\in N_\mu} \frac{\p(\y_k)}{w_k}.
\]
Putting this together with \cref{eq:case1}, we get:
\[ \frac{\p_{-1}(\x''_{b''})}{w_{b''}} \le \frac{\p(\x_\ell)}{w_\ell} \le \frac{\p(\x''_{\ell''})}{w_{\ell''}},\]
where the first inequality uses \cref{eq:case1} and the second inequality uses \cref{lem:le-earning} which shows that the weighted earning of the wLE agent is non-decreasing. The above inequality shows that the wBE does not wpEF1-envy the wLE in $(\x'',\p)$, hence \cref{lem:BEenvyLE} implies that $(\x'',\p)$ must be wpEF1. 
\end{itemize}
Thus in this case, after the last allocation $(\x,\p)$ where $N_1$ was the wBE group, \cref{alg:three-agent-types} halts in at most two subsequent chore transfers, since either allocation $(\x',\p)$ or $(\x'',\p)$ will be wpEF1.

\paragraph{Case 2.} Suppose $j'$ is transferred from $N_1$ to $N_\lambda$ in $(\x,\p)$. We show that $\x'$ is wEF1. Recall that agents belonging to the same group do not wpEF1-envy each other. We analyze the other possibilities below.
\begin{itemize}
\item[(i)] No agent $i\in N_\lambda$ wpEF1-envies another agent $h\in N$ in $(\x',\p)$. This is because:
\[ 
\frac{\p_{-1}(\x'_i)}{w_i} \le \frac{\p(\x'_{\ell'})}{w_{\ell'}} \le \frac{\p(\x'_h)}{w_h},
\]
where the first inequality uses that agents in $N_\lambda$ do not wpEF1-envy each other, and the second inequality uses the fact that $\ell'$ is a wLE agent in $(\x',\p)$.

\item[(ii)] No agent $i\in N_\mu$ wpEF1-envies an agent $h\in N_1$ in $(\x',\p)$. This is because:
\[ 
\frac{\p_{-1}(\x'_i)}{w_i} = \frac{\p_{-1}(\x_i)}{w_i} \le \frac{\p_{-1}(\x_{b})}{w_{b}} \le \frac{\p(\x'_h)}{w_h},
\]
where the equality uses $\x'_i = \x_i$ for $i\in N_\mu$, the first inequality uses that $b\in N_1$ is the wBE agent in $(\x,\p)$, and the second inequality uses \cref{lem:wps-transfer-1} applied to group $N_1$.
\item[(iii)] No agent $i\in N_1$ wpEF1-envies an agent $h\in N_\mu$ in $(\x',\p)$. This is because:
\[
\frac{\p_{-1}(\x'_i)}{w_i} \le \frac{\p_{-1}(\x'_{b'})}{w_{b'}} \le \frac{\p(\x'_h)}{w_h},
\]
where the first inequality uses that $b'$ is the wBE agent in $(\x',\p)$, and the second inequality uses that agents in $N_\mu$ do not wpEF1-each other.

At this point, it only remains to be shown that agents in $N_1$ or $N_\mu$ do not wEF1-envy agents in $N_\lambda$.

\item[(iv)] No agent $i\in N_\mu$ wEF1-envies an agent $\ell\in N_\lambda$. In the following, we refer to the $t^{th}$ step of \cref{alg:three-agent-types} as \textit{time-step} $t$. We use the notation $(\x^t,\p^t)$ to denote the allocation and payments \textit{just after} the execution of the time-step $t$. Thus the initial allocation is $(\x^0,\p^0)$. Let $t_4$ be the time-step such that $(\x^{t_4},\p^{t_4}) = (\x',\p)$.

Let $t_1$ be the most recent time-step in the execution of \cref{alg:three-agent-types} after which $N_\lambda$ becomes the wLE group. Formally, $t_1$ is the largest time-step $t$ s.t. in $(\x^t,\p^t)$, $N_\lambda$ is the wLE group but in $(\x^{(t-1)},\p^{(t-1)})$, $N_\lambda$ is not the wLE group. In other words, $N_\lambda$ is the wLE group during the time-interval $[t_1, t_4]$. Just before time-step $t_1$, the group $N_\mu$ was the wLE group and is the wME group at $t_1$. From \cref{lem:le-change}, we must have that an agent $i\in N_\mu$ (the former wLE group) cannot wpEF1-envy any agent in $N_\lambda$ (the former wME group) at time $t_1$. Thus:
\begin{equation}
\frac{\p_{-1}(\x^{t_1}_i)}{w_i} \le \frac{\p(\x^{t_1}_{\ell_1})}{w_{\ell_1}},
\end{equation}
where $\ell_1\in N_\lambda$ is the wLE agent in $(\x^{t_1},\p^{t_1})$. Thus at time $t_1$, no agent $i\in N_\mu$ wpEF1-envies any agent in $N_\lambda$. From time-step $t_1$ onwards, $N_\lambda$ does not lose any chores as it is the wLE group. Hence it is possible for the agent $i\in N_\mu$ to start wEF1-envying an agent in $N_\lambda$ only if $i$ gains a chore. Let $t_2\in[t_1,t_4]$ be the last time-step during which a chore $j_1$ is transferred from $N_1$ to $N_\mu$. Let $t_2^- := t_2-1$ be the previous time-step. Let $M_\mu = \bigcup_{i\in N_\mu} \x_i^{t_2^-}$, and $\mpb_\lambda$ denote the MPB set of $N_\lambda$ at $(\x^{t_2^-}, \p^{t_2^-})$. Because \cref{alg:three-agent-types} performed the transfer of $j_1$ from $N_1$ to $N_\mu$, it must be that $|M_\mu \cap \mpb_\lambda|\neq\emptyset$ and for every $j_2\in M_\mu \cap \mpb_\lambda$, we have for allocation $\y = \wps(N_\mu, M_\mu\setminus j_2)$ as computed in Line 16, we have:
\begin{equation}\label{eq:case2}
\min_{k\in N_\mu} \frac{\p^{t_2^-}(\y_k)}{w_k } \le \frac{\p^{t_2^-}(\x^{t_2^-}_\ell)}{w_\ell}, \text{for all $\ell\in N_\lambda$}.    
\end{equation} 

We now argue that if agent $i\in N_\mu$ wpEF1-envies an agent in $N_\lambda$ in the allocation $(\x^{t_2}, \p^{t_2})$, then \cref{alg:three-agent-types} will perform a chore transfer from $N_\mu$ to $N_\lambda$ at some later time-step $t_3\in[t_2,t_4]$ such that in $(\x^{t_3}, \p^{t_3})$ no agent in $N_\mu$ wpEF1-envies any agent in $N_\lambda$.

To see why, let us examine the chore transfers occuring in $[t_2, t_4]$. By definition of $t_2$, there are no $N_1$ to $N_\mu$ chore transfers in this interval. If there are also no $N_\mu$ to $N_\lambda$ chore transfers in this interval, then at time-step $t_4$, group $N_1$ ceases to be the wBE group due to a chore transfer from $N_1$ to $N_\lambda$. Hence by the analysis in Case (1), we have that $(\x^{t_4}, \p^{t_4})$ is wpEF1.

Let us therefore consider the earliest time-step $t_3\in[t_2,t_4)$ when a $N_\mu$ to $N_\lambda$ transfer takes place. With $t_3^-:= t_3-1$ and the definitions of $t_2$ and $t_3$, one notes that $\x^{t_3^-}_i = \x^{t_2}_i$ for every $i\in N_\mu$. By defining $\hat{\x}^{t_3}$ as the allocation $\x^{t_3}$ restricted to agents in $N_\mu$, we observe:

\begin{equation*}
\begin{aligned}
\hat{\x}^{t_3} &= \wps(N_\mu, \bigcup_{i\in N_\mu} \x_i^{t_3^-}\setminus j_2) \\
&= \wps(N_\mu, \bigcup_{i\in N_\mu} \x_i^{t_2}\setminus j_2) \\
&= \wps(N_\mu, (M_\mu\setminus j_2) \cup j_1).
\end{aligned}
\end{equation*}

Comparing this with $\y = \wps(N_\mu, M_\mu\setminus j_2)$, we can apply \cref{lem:wps-transfer-1} to the group $N_\mu$ to conclude for any $i\in N_\mu$:
\[ 
\frac{\p_{-1}^{t_2^-}(\x^{t_3}_i)}{w_i} \le \min_{k\in N_\mu}\frac{\p^{t_2^-}(\y_k)}{w_k}.
\]
Using this with \cref{eq:case2} we get that:
\[ 
\frac{\p_{-1}^{t_2^-}(\x^{t_3}_i)}{w_i} \le \min_{\ell\in N_\lambda}\frac{\p^{t_2^-}(\x^{t_2^-}_\ell)}{w_\ell}.
\]
Now notice that in the period $[t_2, t_3]$, the groups $N_\lambda$ and $N_\mu$ undergo payment drops together. Together with the fact that $\x^{t_2^-}_\ell \subseteq \x^{t_3}_\ell$ for all $\ell\in N_\lambda$, we also have:
\[ 
\frac{\p_{-1}^{t_3}(\x^{t_3}_i)}{w_i} \le \min_{\ell\in N_\lambda}\frac{\p^{t_3}(\x^{t_3}_\ell)}{w_\ell},
\]
showing that no agent $i\in N_\mu$ wpEF1-envies any agent in $N_\lambda$ at $(\x^{t_3}, \p^{t_3})$. Subsequent to $t_3$, the group $N_\mu$ does not gain any new chores and the group $N_\lambda$ will not lose any chores. Hence no agent in $N_\mu$ wEF1-envies any agent in $N_\lambda$ at $(\x^{t_4}, \p^{t_4})$ as well.

\item[(v)] We now argue that no agent $i\in N_1$ wEF1-envies an agent $\ell\in N_\lambda$ at $(\x',\p) = (\x^{t_4}, \p^{t_4})$. Observe that:
\begin{equation}\label{eq:case2v}
\begin{aligned}
\frac{\p_{-1}^{t_4}(\x^{t_4}_i)}{w_i} &\le \frac{\p_{-1}^{t_4}(\x^{t_4}_b)}{w_b} \\
&\le \max_{h\in N_\mu}\frac{\p_{-1}^{t_3}(\x^{t_3}_h)}{w_h} \\
&\le \min_{\ell\in N_\lambda}\frac{\p^{t_3}(\x^{t_3}_\ell)}{w_\ell},    
\end{aligned}
\end{equation}
where the first inequality uses that $b\in N_\mu$ is the wBE agent at $(\x^{t_4},\p^{t_4})$, the second inequality uses the fact that in the interval $[t_3, t_4]$, the group $N_\mu$ either undergoes payment drops or loses chores, but does not gain any chores, and the third inequality uses the fact from part (iv) that no agent $h\in N_\mu$ wpEF1-envies any agent $\ell\in N_\lambda$ at $(\x^{t_3}, \p^{t_3})$.

Now observe that no agent in $N_1$ ever undergoes a payment drop, hence the MPB ratio of all agents in $N_1$ is one. \cref{eq:case2v} therefore implies:
\begin{equation}\label{eq:case2v2}
\begin{aligned}
\min_{j\in \x^{t_4}_i} \frac{d_i(\x^{t_4}_i\setminus j)}{w_i} &= \frac{\p_{-1}^{t_4}(\x^{t_4}_i)}{w_i} \\
&\le \min_{\ell\in N_\lambda}\frac{\p^{t_3}(\x^{t_3}_\ell)}{w_\ell} \\
&\le \min_{\ell\in N_\lambda}\frac{d_i(\x^{t_3}_\ell)}{w_\ell}.
\end{aligned}
\end{equation}

Moreover since an agent $N_\lambda$ can only gain a chore in the interval $[t_3, t_4]$, we have $\x^{t_3}_\ell \subseteq \x^{t_4}_\ell$ for all $\ell\in N_\lambda$. Together with \cref{eq:case2v2} this implies that:
\begin{equation}
\min_{j\in \x^{t_4}_i} \frac{d_i(\x^{t_4}_i\setminus j)}{w_i} \le \min_{\ell\in N_\lambda}\frac{d_i(\x^{t_3}_\ell)}{w_\ell} \le \min_{\ell\in N_\lambda}\frac{d_i(\x^{t_4}_\ell)}{w_\ell},
\end{equation}
showing that in the allocation $\x^{t_4}$, no agent $i\in N_1$ wEF1-envies any agent $\ell\in N_\lambda$.
\end{itemize}
Since these cases are exhaustive, we conclude that the allocation $\x'$ is wEF1.
\end{proof}

\section{wEF1+fPO for Two-Chore-Types Instances}\label{sec:two-chore-types}
We now turn to the problem of existence and computation of a wEF1+fPO allocation for instances with two types of chores. We answer this question positively in this section. Thus, our result generalizes that of \cite{aziz2022twotypes}, who showed that EF1+fPO allocations can be computed in polynomial time for two-chore-types instances with unweighted agents. 

\begin{theorem}\label{thm:two-chore-types}
In any two chore type instance, a wEF1 and fPO allocation exists, and can be computed in polynomial time.
\end{theorem}

To prove \cref{thm:two-chore-types}, we present a polynomial time algorithm, \cref{alg:two-chore-types}, that computes a wEF1 and fPO allocation for a given two-chore-types instance. We note that \cref{alg:two-chore-types} is also an alternative algorithm to that of \cite{aziz2022twotypes} for computing an EF1+PO allocation.

\begin{algorithm}[!t]
\caption{wEF1+fPO for two-chore-types instances}\label{alg:two-chore-types}
\textbf{Input:} Two-chore-types instance $(N,W,M,D)$\\
\textbf{Output:} An integral wEF1 and fPO allocation $\x$
\begin{algorithmic}[1]
\State Let $M = A\sqcup B$ be a partition of chores according to type
\For{$i=1$ to $n$}
\State $\x_i \gets M$, $\x_h \gets \emptyset$ for $h\neq i$
\State Set payments as $p_A = d_{iA}$ and $p_B = d_{iB}$
\While{$(\x,\p)$ is not wpEF1}
    \State $L \gets \{k: k\in \arg\min_{h \in [n]} \frac{\p(\x_h)}{w_h}\}$ \Comment{$L$ is the set of global wLE agents}
    \State $\ell_A \gets \max (L \cap [i-1])$ \Comment{Set $\ell_A$ as the maximum index wLE in $[i-1]$; $\max\emptyset=\infty$}
    \State $\ell_B \gets \min (L\cap [n]\setminus [i])$ \Comment{Set $\ell_B$ as the minimum index wLE in $[n]\setminus[i]$; $\min\emptyset=0$}
    \If{$\ell_A < i$ and $\x_i \cap A \neq \emptyset$} \Comment{Transfer $A$-chore from $i$ to $\ell_A$}
        \State Let $a \in \x_i \cap A$ be an $A$-chore in $\x_i$
        \State $\x_i \gets \x_i\setminus a$, $\x_{\ell_A} \gets \x_{\ell_A} \cup a$ 
    \ElsIf{$\ell_B > i$ and $\x_i \cap B \neq \emptyset$} \Comment{Transfer $B$-chore from $i$ to $\ell_B$}
        \State Let $b \in \x_i \cap B$ be a $B$-chore in $\x_i$
        \State $\x_i \gets \x_i\setminus b$, $\x_{\ell_B} \gets \x_{\ell_B} \cup b$ 
    \Else: Break \Comment{Failure in Phase $i$}
    \EndIf
\EndWhile
\If{$(\x,\p)$ is wpEF1}
\State \Return $\x$
\EndIf
\EndFor
\end{algorithmic}
\end{algorithm}

Let $M = A \sqcup B$ be a partition of the chores into two sets, each containing chores of the same type. For $X\in\{A,B\}$, we refer to a chore in set $X$ as an $X$-chore, and denote by $d_{iX}$ the disutility an agent $i$ has for any chore in set $X$. 

\subsection{Algorithm Description}
We first sort and re-label the agents in non-decreasing order of $d_{iA}/d_{iB}$, i.e., for $i < j$, $\frac{d_{iA}}{d_{iB}} \le \frac{d_{jA}}{d_{jB}}$. Roughly speaking, agents with a smaller index prefer to do $A$-chores over $B$-chores, and vice-versa. 

Our algorithm proceeds in at most $n$ phases. Starting from $i=1$, in Phase $i$ we select agent $i$ as the \textit{pivot} agent. The pivot agent $i$ is initially assigned all the chores (making $i$ the only agent with wpEF1-envy and the unique wBE), and the payments of the chores are set according to the disutilities of the pivot. In other words, we set the payment of a chore in set $X$ as $p_X = d_{iX}$ for $X\in\{A,B\}$.  While the allocation is not wpEF1, we attempt to reduce the wpEF1-envy of $i$ by transferring a chore from $i$ to a wLE agent (this maintains the property that while the allocation is not wpEF1, $i$ is the only agent with wpEF1-envy and is the unique wBE). Let $\ell_A$ be the index of a wLE agent among agents $1$ through $(i-1)$ if such an agent exists (otherwise let $\ell_A = \infty$), with ties broken in favour of larger index (Line 7). Similarly, let $\ell_B$ be the index of a wLE agent among agents $(i+1)$ through $n$ if such an agent exists (otherwise let $\ell_B = 0$), with ties broken in favour of smaller index (Line 8). We attempt to make a chore transfer as follows:
\begin{itemize}[leftmargin=*]
    \item[] (Lines 9-11) If $\ell_A < i$ and $\x_i \cap A \neq \emptyset$, then there is a wLE agent to the left of the pivot agent $i$ and $i$ has an $A$-chore which is on MPB for $\ell_A$. Thus, we transfer this $A$-chore from $i$ to $\ell_A$. 
    \item[] (Lines 12-14) If $\ell_B > i$ and $\x_i \cap B \neq \emptyset$, then there is a wLE agent to the right of pivot agent $i$ and $i$ has a $B$-chore which is on MPB for $\ell_B$. Thus, we transfer this $B$-chore from $i$ to $\ell_B$.
    \item[] (Line 15) If neither of the above cases are met, then we have encountered a failure in Phase $i$. It cannot be that a wLE exists on both sides of $i$, as then $i$ must have been able to transfer some chore to a wLE agent (since we maintain that $i$ is the wBE her bundle must be non-empty). This implies that either (i) a wLE agent exists only on the left side of $i$ but $i$ has no $A$-chores to transfer, or (ii) a wLE agent exists only on the right side of $i$ but $i$ has no $B$-chores to transfer. In the case of (i), we say that agent $i$ faces an $A$-fail in Phase $i$, and in the case of (ii) we say that agent $i$ faces a $B$-fail in Phase $i$. After facing a failure, no more transfers are performed in Phase $i$.
\end{itemize}
\cref{alg:two-chore-types} terminates either when a wpEF1 allocation is found, or when all phases are completed.

\subsection{Analysis of \cref{alg:two-chore-types}}
We now prove \cref{thm:two-chore-types} by arguing that \cref{alg:two-chore-types} finds a wpEF1 and fPO allocation. We first show:
\begin{restatable}{lemma}{lemTwoChoreFpo}\label{lem:two-chore-types-fpo}
Throughout the run of \cref{alg:two-chore-types}, every allocation is fPO.
\end{restatable}
\begin{proof}
Consider Phase $i$. Recall that the payment of a chore $j\in X$ is $p_X = d_{iX}$, for $X\in \{A,B\}$. Consider an allocation $\x$ during Phase $i$. By the nature of transfers performed by \cref{alg:two-chore-types}, we have $\x_h \subseteq A$ for all $h<i$, and $\x_h \subseteq B$ for all $h>i$. We show $(\x,\p)$ is on MPB. This is clear for agent $i$. For an agent $h < i$, the ordering of the agents implies that $\frac{d_{hA}}{d_{hB}} \le \frac{d_{iA}}{d_{iB}}$. Thus, $\frac{d_{hA}}{p_A} \le \frac{d_{hB}}{p_B}$, showing that $\x_h$ is on MPB for agent $h$. Likewise for an agent $h > i$, the ordering of the agents implies that $\frac{d_{iA}}{d_{iB}} \le \frac{d_{hA}}{d_{hB}}$. Thus, $\frac{d_{hB}}{p_B} \le \frac{d_{hA}}{p_A}$, once again showing that $\x_h$ is on MPB for agent $h$. By the Second Welfare Theorem we can conclude that $\x$ is fPO since $(\x, \p)$ is on MPB.
\end{proof}

To show that \cref{alg:two-chore-types} eventually finds a wpEF1 allocation, we first show:
\begin{restatable}{lemma}{lemienvy}\label{lem:i-envy}
Throughout the execution of Phase $i$, if the allocation is not wpEF1 then agent $i$ is the only weighted big earner.
\end{restatable}
\begin{proof}
We prove this inductively. In the initial allocation where all chores are assigned to $i$, it is clear that $i$ is the only agent with wpEF1-envy and is the wBE. Suppose it is true that in an allocation $\x$ during the execution of Phase $i$, agent $i$ is the only agent with wpEF1-envy. Let $\x'$ be the allocation obtained by transferring one chore $j$ from $i$ to a wLE $\ell$. 
Thus $\x'_i = \x_i \setminus j$, $\x'_\ell = \x_\ell \cup j$, and for $h \notin \{i, \ell\}$, we have $\x'_h = \x_h$. 
We show that no agent $h\neq i$ has wpEF1-envy towards any other agent. We have the following four cases to consider:
\begin{itemize}[leftmargin=*]
\item Agents $h\notin\{i,\ell\}$ does not wpEF1-envy $k\neq i$. Note that $\x'_h = \x_h$ and $\x_k \subseteq \x'_k$. Since $h$ did not wpEF1-envy $k$ in $(\x,\p)$, $h$ will not wpEF1-envy agent $k$ in $(\x',\p)$ as well. 
\item Agent $h\notin\{i,\ell\}$ does not wpEF1-envy $i$ in $\x'$. Observe the following.
\[ \frac{\p_{-1}(\x'_h)}{w_h} = \frac{\p_{-1}(\x_h)}{w_h} \le \frac{\p_{-1}(\x_i)}{w_i} \le \frac{\p(\x_i \setminus j)}{w_i} = \frac{\p(\x'_i)}{w_i}.\]
Here: (i) the first equality holds because $\x'_h = \x_h$, (ii) the first inequality holds because $i$ is a wBE in $(\x,\p)$, (iii) the second inequality is due to the definition of $\p_{-1}(\cdot)$, and (iv) the last equality is because $\x'_i = \x_i \setminus j$.
\item Agent $\ell$ does not wpEF1-envy any agent $h\notin \{i, \ell\}$. Observe the following.
\[ \frac{\p_{-1}(\x'_\ell)}{w_\ell} \le \frac{\p(\x_\ell)}{w_\ell} \le \frac{\p(\x_h)}{w_h} = \frac{\p(\x'_h)}{w_h} .\]
Here: (i) the first inequality uses $\x'_\ell = \x_\ell \cup j$ and the definition of $\p_{-1}(\cdot)$, (ii) the second inequality holds because $\ell$ is a wLE in $(\x,\p)$, and (iii) the last equality uses $\x'_h = \x_h$.
\item Agent $\ell$ does not wpEF1-envy agent $i$. Observe the following.
\[ \frac{\p_{-1}(\x'_\ell)}{w_\ell} \le \frac{\p(\x_\ell)}{w_\ell} < \frac{\p_{-1}(\x_i)}{w_i} \le \frac{\p(\x'_i)}{w_i} .\]
Here: (i) the first inequality uses $\x'_\ell = \x_\ell \cup j$ and the definition of $\p_{-1}(\cdot)$, (ii) the second inequality uses the fact that $(\x,\p)$ was not wpEF1, and (iii) the last equality uses $\x'_i = \x_i \setminus j$.
\end{itemize}
Thus, only agent $i$ can wpEF1-envy another agent in $(\x', \p)$. Consequently, $i$ is the only wBE.
\end{proof}

The above lemma shows that if we did not find a wpEF1 allocation in Phase $i$, it must be because the only wBE $i$ could not transfer a chore to a wLE agent. That is, Phase $i$ was terminated due to $i$ facing either an $A$-fail or a $B$-fail. We next prove that:
\begin{restatable}{lemma}{lemfailures}\label{lem:failures}
If agent $i$ faces a $B$-fail, then agent $(i+1)$ cannot face an $A$-fail.
\end{restatable}
Before proving \cref{lem:failures}, we note that we can now prove \cref{thm:two-chore-types}. By definition of a failure, agent $1$ cannot face an $A$-fail. Thus if agent $1$ faces a failure, it must be due to a $B$-fail. By \cref{lem:failures}, we know that agent $2$ cannot face an $A$-fail. Proceeding inductively, we conclude for each $i\in[n]$ that either agent $i$ faces a $B$-fail or does not face a failure at all. However, agent $n$ cannot face a $B$-fail by definition of failure. Thus it must be that there is some pivot agent $i^*\in[n]$ who did not face a failure. 

Therefore, by invoking \cref{lem:i-envy} we conclude that in Phase $i^*$, it was always possible to transfer chores from the \textit{only} wBE agent $i^*$ to a wLE agent, until the allocation became wpEF1. \cref{lem:two-chore-types-fpo} shows that all allocations encountered in \cref{alg:two-chore-types} are fPO. Thus, a wEF1 and fPO allocation will be found in Phase $i^*$. Since there are at most $n$ phases and there are at most $m$ chore transfers in each phase, \cref{alg:two-chore-types} terminates in polynomial time. This proves \cref{thm:two-chore-types}.

It only remains to prove the crucial \cref{lem:failures}. To do this, we 
relate the allocation returned by \cref{alg:two-chore-types} in Phase $i$ with the allocation returned by running the weighted picking sequence algorithm for assigning $A$-chores to $[i]$ and $B$-chores to $[n]\setminus [i]$. Let $\A$ be the algorithm given by \cref{alg:A}, which allocates identical chores to agents by initially assigning all chores to agent 1 and then transferring chores one by one to the smallest index wLE agent until the allocation is wpEF1. Thus $\A$ mimics the execution of Phase $i$ of \cref{alg:three-agent-types} for assigning $B$-chores, as well as $A$ chores but in reverse order. Our next lemma proves that the executions of $\A$ and $\wps$ return the same allocation for identical chores.

\begin{algorithm}[t]
\caption{Algorithm $\A$}\label{alg:A}
\textbf{Input:} Agents $N$, set of identical chores $M$ \\
\textbf{Output:} A wEF1 allocation $\x$
\begin{algorithmic}[1]
\State $\x_1 \gets M$; $\x_i \gets \emptyset$ for $i\neq 1$ 
\State Set payments as $p_j = 1$ for all $j\in M$
\While{$(\x,\p)$ is not wpEF1}
\State $i \gets \arg\min_{k\ge 2} \frac{\p(\x_k)}{w_k}$; ties broken in favour of smaller index
\State Transfer a single chore from $\x_1$ to $\x_i$
\EndWhile
\State \Return $\x$
\end{algorithmic}
\end{algorithm}

\begin{lemma}\label{lem:two-chore-type-wps}
Consider a set $N$ of agents and a set $M$ of identical chores. Then the allocations returned by \cref{alg:A} and $\wps$ are the same. That is, $\A(N, M) = \wps(N, M)$.
\end{lemma}
\begin{proof}
Let $\x(t)$ (resp. $\y(t)$) denote the allocation after $t$ iterations of $\A(N,M)$ (resp. $\wps(N, M)$). Since the chores are identical, an allocation is characterized simply by the number of chores assigned to each agent. Let $x_i(t)$ (resp. $y_i(t)$) denote the number of chores assigned to agent $i$ in $\x(t)$ (resp. $\y(t)$). Note that $\wps(N,M)$ runs for $m$ iterations. We claim that:
\begin{quote}
For all $0\le t\le m$ and $2\le i\le n$, $y_i(t) = x_i(t-y_1(t))$.
\end{quote}
We prove this claim by induction on $t$. For $t=0$ the claim holds since $y_i(0) = x_i(0) = 0$ for all $i\neq 1$. Suppose the claim is true for some $0\le t\le m-1$. Consider iteration $(t+1)$ of $\wps(N,M)$. We have two cases.
\begin{itemize}[leftmargin=*]
\item $\wps(N,M)$ assigns a chore to agent $1$ in iteration $(t+1)$. Then $y_1(t+1) = y_1(t)+1$ and $y_i(t+1) = y_i(t)$ for $i\neq 1$ . We then have for $i\neq 1$:
\[
x_i(t+1 - y_1(t+1)) = x_i(t-y_1(t)) = y_i(t) = y_i(t+1),
\]
where we used the induction hypothesis in the second equality.
\item $\wps(N,M)$ assigns a chore to agent $k\neq 1$ in iteration $(t+1)$. We argue that $k$ is a minimum-index wLE in $(\x(t-y_1(t)),\p)$; recall that $\A$ sets all payments equal to $1$. Since $\wps$ chose agent $k$ to be assigned a chore in iteration $(t+1)$, agent $k$ must be a minimum-index agent satisfying $\frac{y_k(t)}{w_k} \le \frac{y_h(t)}{w_h}$, for all $h\in N$. Then by the induction hypothesis we have $\frac{x_k(t-y_1(t))}{w_k} \le \frac{x_h(t-y_1(t))}{w_h}$, for all $h\in N$, implying $k$ is the minimum-index wLE in $(\x(t-y_1(t)), \p)$. 

We show that agent $1$ wpEF1-envies agent $k$ in $(\x(t-y_1(t)),\p)$ as follows:
\begin{align*}
&\frac{x_k(t-y_1(t))}{w_k} = \frac{y_k(t)}{w_k} \qquad\qquad\qquad\qquad\quad\;\;{\text{ (using IH)}}\\
&< \frac{y_1(t)}{w_1} \qquad\qquad\qquad\quad{\text{ (since $k$ is the min-index wLE)}} \\
&= \frac{t-\sum_{i\ge 2} y_i(t)}{w_1} \qquad\qquad\qquad\quad\quad\text{ ($\y(t)$ has $t$ chores)}\\
&\le \frac{t-\sum_{i\ge 2} x_i(t-y_1(t))}{w_1} \qquad\qquad\qquad\quad\quad\text{ (using IH)} \\
&= \frac{m-1-\sum_{i\ge 2} x_i(t-y_1(t))}{w_1} \qquad\;\;\text{ (since $t\le m-1$)} \\
&= \frac{x_1(t-y_1(t))}{w_1}. \;\;\text{(since $\x(t-y_1(t))$ has $m$ total chores)}    
\end{align*}
    
Thus $\A$ will perform a transfer of a chore from agent $1$ to agent $k$ in $(\x(t-y_1(t)),\p)$. This results in the allocation $\x(t+1-y_1(t))$, which equals $\x(t+1-y_1(t+1))$ since $y_1(t+1) = y_1(t)$. Using the induction hypothesis, the claim then follows by noting that for agent $i\in[n]\setminus\{1,k\}$:
\begin{equation*}
\begin{aligned}
x_i(t+1-y_1(t+1)) &= x_i(t+1-y_1(t)) \\
&=x_i(t-y_1(t) \\
&= y_i(t) = y_i(t+1).
\end{aligned}
\end{equation*}
and for agent $k$:
\begin{equation*}
\begin{aligned}
x_k(t+1-y_1(t+1)) &= x_k(t+1-y_1(t)) \\
&= x_k(t-y_1(t))+1 \\
&= y_k(t) + 1 = y_k(t+1).  
\end{aligned}
\end{equation*}
\end{itemize}

Thus by induction the claim holds. Plugging $t=m$ we obtain $x_i(m-y_1(m)) = y_i(m)$ for all $i\neq 1$. Then we have $x_1(m-y_1(m)) = m-\sum_{i\ge 2} x_i(m-y_1(m)) = m - \sum_{i\ge 2} y_i(m) = y_1(m)$. This shows that $\x(m-y_1(m)) = \y$. Since $\y=\wps(N,M)$ is wEF1, so is the allocation $\x(m-y_1(m))$, and $\A(N,M)$ halts and returns with $\x(m-y_1(m))$. This proves that $\A(N,M) = \wps(N,M)$.
\end{proof}

We now prove \cref{lem:failures}.
\lemfailures*
\begin{proof}
For the sake of contradiction, assume agent $i$ faces a $B$-fail and agent $(i+1)$ faces an $A$-fail. Let $(\x, \p)$ be the allocation when Phase $i$ was terminated and let $\ell > i$ be the minimum-index wLE. Moreover, all chores in $A$ are allocated to agents $1$ to $i$ while chores in $B$ are allocated to agents $(i+1)$ to $n$. 
Next, let $(\x', \p')$ be the allocation when Phase $(i+1)$ was terminated and let $k < i+1$ be the maximum index wLE. As in $\x$, in $\x'$ all chores in $A$ are allocated to agents $1$ to $i$ while chores in $B$ are allocated to agents $(i+1)$ to $n$.

Let us first examine the allocation of $B$ to agents $(i+1)$ to $n$. Let $\hat{\x}$ (resp. $\hat{\x}'$) be the allocation $\x$ (resp. $\x'$) when restricted to agents $(i+1)$ to $n$. Observe that $\hat{\x}$ is obtained by iteratively allocating a $B$-chore to the wLE in $[n]\setminus [i]$. Since the payment of each $B$-chore is the same, the wLE is the agent with the least ratio $t_\ell/w_\ell$, where $t_\ell$ is the number of $B$-chores allocated to agent $\ell$. Moreover since transfers are made to minimum-index wLE agents, the allocation $\hat{\x}$ is the same as $\wps([n]\setminus[i], B)$, the allocation obtained by using the weighted picking sequence algorithm applied to agents $(i+1)$ to $n$ with ties broken in favour of smaller index agents. Thus, $\hat{\x} = \wps([n]\setminus [i], B)$. On the other hand, $\hat{\x}'$ is obtained by initially allocating all chores in $B$ to agent $(i+1)$ and then iteratively transferring a chore to the minimum-index wLE. Thus, $\hat{\x}'$ is an allocation encountered in the run of algorithm $\A$ (Alg. \ref{alg:A}) applied to agents $(i+1)$ to $n$, with all chores initially assigned to $(i+1)$. Thus we have for all $h \in (i+1,n]$, $\hat{\x}'_h\subseteq \y_h$, where $\y = \A([n]\setminus[i], B)$. In \cref{lem:two-chore-type-wps} we show that $\wps([n]\setminus [i], B) = \A([n]\setminus[i], B)$. Thus we have for $h\in (i+1, n]$, $\hat{\x}'_h \subseteq \hat{\x}_h$. In particular, this implies $\x'_\ell \subseteq \x_\ell$.

Let us next examine the allocation of $A$ to agents $1$ to $i$. Let $\tilde{\x}$ (resp. $\tilde{\x}'$) be the allocation $\x$ (resp. $\x'$) when restricted to agents $1$ to $i$. Note that in $\x$, chores are transferred from $i$ to the maximum index wLE. Hence, by reversing the labels of agents in $[i]$ and using similar arguments, we have that $\tilde{\x}'$ is the same as $\wps([i], A)$, and $\tilde{\x}$ is an allocation encountered in the run of $\A$ applied to agents $1$ to $i$. Thus we have for all $h\in[1,i)$, $\tilde{\x}_h \subseteq \z_h$, where $\z = \A([i], A)$. \cref{lem:two-chore-type-wps} shows $\wps([i], A) = \A([i], A)$, hence $\tilde{\x}_h \subseteq \tilde{\x}'_h$ for $h\in[1,i)$. In particular, this implies $\x_k \subseteq \x'_k$.

We can now complete the proof of the lemma. Observe the following:
\begin{equation}\label{eq:A-fail-ineq}
\frac{\p'(\x_k)}{w_k} \le \frac{\p'(\x'_k)}{w_k} \le \frac{\p'(\x'_\ell)}{w_\ell} \le \frac{\p'(\x_\ell)}{w_\ell},
\end{equation}
where the first inequality uses $\x_k \subseteq \x'_k$, the second uses that $k$ is a wLE in $(\x',\p')$, and the third uses $\x'_\ell \subseteq \x_\ell$. Next, we note that since $\ell$ is the wLE in $(\x,\p)$ and $k$ is not, we have:
\begin{equation}\label{eq:B-fail-ineq}
\frac{\p(\x_k)}{w_k} > \frac{\p(\x_\ell)}{w_\ell},
\end{equation}

Using the above two inequalities \eqref{eq:B-fail-ineq} and \eqref{eq:A-fail-ineq}, we get $\frac{\p(\x_\ell)}{\p(\x_k)} < \frac{w_\ell}{w_k} \leq \frac{\p'(\x_\ell)}{\p'(\x_k)}$. Note that $\p(\x_\ell) = |\x_\ell|\cdot p_B = |\x_\ell|\cdot d_{iB}$ and $\p(\x_k) = |\x_k|\cdot p_A = |\x_\ell|\cdot d_{iA}$. Similarly, $\p'(\x_\ell) = |\x_\ell|\cdot p'_B = |\x_\ell|\cdot d_{(i+1)B}$ and $\p'(\x_k) = |\x_k|\cdot p'_A = |\x_k|\cdot d_{(i+1)A}$. Putting these together, we obtain $\frac{d_{iB}}{d_{iA}} < \frac{d_{(i+1)B}}{d_{(i+1)A}}$, which is a contradiction, since we sorted agents in non-decreasing order of $d_{iA}/d_{iB}$.
\end{proof}

\section{Discussion}
In this work we studied the problem of computing a wEF1 and fPO allocation of chores for agents with unequal weights or entitlements. We showed positive algorithmic results for this problem for instances with three types of agents, or two types of chores. The existence of EF1 and fPO allocations of chores in the symmetric case remains a hard open problem. Our results further our understanding of this problem by contributing to the body of positive non-trivial results. Together with the result of \cite{wu2023wef1} concerning bivalued chores, our paper shows that wEF1 and fPO allocations exists for every structured instance known so far that admits an EF1 and fPO allocation. Our idea of combining the competitive equilibrium framework with envy-resolving algorithms like the weighted picking sequence algorithms could be an important tool in settling the problem in its full generality.

\appendix
\section{Examples}~\label{app:examples}
\begin{example}
An unweighted EF1+PO allocation in an instance with copies of agents does not give a weighted wEF1+PO allocation.
\end{example}
Consider an instance with two agents $a$ and $b$ and three chores, with $w_a = 1$ and $w_b = 2$, and disutilities given by:
\begin{center}
\begin{tabular}{ |c||c|c|c| } 
 \hline
  & $j_1$ & $j_2$ & $j_3$ \\
 \hline
 $a$ & 1 & 10 & 10 \\ 
 \hline
 $b$ & 2 & 10 & 10 \\ 
 \hline
\end{tabular}
\end{center}

Suppose we create two copies of agent $b$: agents $b_1$ and $b_2$, and construct an EF1+PO allocation in the resulting unweighted instance with three agents $a, b_1$ and $b_2$. Observe that the following allocation $\x$ is EF1+PO:
\begin{itemize}
    \item $\x_a = \{j_1\}$,
    \item $\x_{b_1} = \{j_2\}$,
    \item $\x_{b_2} = \{j_3\}$.
\end{itemize}
Now consider the allocation $\y$ where agents $b_1$ and $b_2$ are combined to give $\y_b = \y_{b_1} \cup \y_{b_2}$, and $\y_a = \x_a$. However $\y$ is not wEF1 since $b$ wEF1 envies $a$ as $d_b(\y_b\setminus j)/w_b = 10/2 > 2 = d_b(\y_a)/w_a$ for $j\in\{j_2,j_3\}$.

\begin{example}
A weighted wEF1+PO allocation does not give an unweighted EF1+PO allocation with copies of agents. 
\end{example}
We observe the following weighted instance with two agents $a$ and $b$ and seven chores, where $w_a = 1$ and $w_b = 3$.
\begin{center}
\begin{tabular}{ |c||c|c|c|c|c|c|c| } 
 \hline
  & $j_1$ & $j_2$ & $j_3$ & $j_4$ & $j_5$ & $j_6$ & $j_7$\\
 \hline
 $a$ & 2 & 2 & 2 & 3 & 3 & 3 & 3 \\ 
 \hline
 $b$ & 100 & 100 & 100 & 99 & 99 & 99 & 99 \\ 
 \hline
\end{tabular}
\end{center}
Then, the following allocation $\x$ is wEF1+PO:
\begin{itemize}
    \item $\x_a = \{j_1, j_2, j_3\}$,
    \item $\x_b = \{j_4, j_5, j_6, j_7\}$.
\end{itemize}
However, this allocation does not give an unweighted EF1+PO allocation where agent $b$ is split into three copies $b_1$, $b_2$, and $b_3$, and $\x_{b_1} \cup \x_{b_2} \cup \x_{b_3} = \x_b$. Indeed, some copy of $b$ must receive at most one chore, so this agent will be EF1-envied by $a$.

\newpage
\bibliography{references}

\end{document}